\documentclass[11pt,a4paper]{article}
\usepackage[margin=1in]{geometry}

\usepackage[font=bera]{chihao}

\usepackage{tocloft}
\usepackage{booktabs}
\usepackage{threeparttable}
\usepackage{multirow}
\usepackage{array}
\usepackage{makecell}

\usepackage{catchfilebetweentags}
\usepackage{relsize}

\newcommand{\DTV}{\!{TV}}
\newcommand{\DKL}[2]{\!{KL}\tp{#1\,\|\,#2}}
\newcommand{\OU}{\!{OU}}
\newcommand{\Id}{\!{Id}_d}
\newcommand{\RGO}{\!{RGO}}
\newcommand{\Rmnum}[1]{\uppercase\expandafter{\romannumeral #1}}
\crefname{assumption}{Assumption}{Assumptions}
\Crefname{assumption}{Assumption}{Assumptions}
\crefname{condition}{Condition}{Conditions}
\Crefname{condition}{Condition}{Conditions}

\title{Improved sampling algorithms and functional inequalities for non-log-concave distributions}

\author{Yuchen He \\ Shanghai Jiao Tong University \\ \textsf{yuchen\_he@sjtu.edu.cn} \and Zhehan Lei \\ Shanghai Jiao Tong University \\ \textsf{hsfzlzh1@sjtu.edu.cn} \and Jianan Shao \\ Shanghai Jiao Tong University \\ \textsf{sjtu13362955623@sjtu.edu.cn}  \and Chihao Zhang\\ Shanghai Jiao Tong University \\ \textsf{chihao@sjtu.edu.cn}}

\begin{document}
\maketitle
% \ctodo{Tentative title}
 
\begin{abstract}

We study the problem of sampling from a distribution $\mu$ with density proportional to $e^{-V}$ for a potential function $V:\bb R^d\to \bb R$ with query access to $V$ and $\grad V$. We begin with two standard assumptions on $\mu$:
\begin{itemize}
    \item[](1) The potential function $V$ is $L$-smooth.
    \item[](2) The second moment of $\mu$ is finite, \IE, $\E[X\sim \mu]{\|X\|^2}\leq M$ for some $M<\infty$.
\end{itemize}
Recently, He and Zhang (COLT, 2025) showed that the  query complexity for sampling from this family of distributions can be as large as $\tp{\frac{LM}{d\eps}}^{\Omega(d)}$, where $\eps$ is the target accuracy in total variation distance, and this indicates that the \Poincare constant can be arbitrarily large.%\footnote{Otherwise, a bounded \Poincare constant can imply samplers with $\!{poly}(1/\eps)$ queries (Chewi, Erdogdu, Li, Shen and Zhang (FoCM, 2024)), contradicting the lower bound in the work of He and Zhang (COLT, 2025).}. 

On the other hand, another common assumption in the literature on diffusion-based sampling algorithms (see \EG, the work of Chen, Chewi, Li, Li, Salim and Zhang (ICLR, 2023)) strengthens the smoothness condition (1) on $V$ to the following:
\begin{itemize}
    \item [](1*) The potential function of \emph{every} distribution along the Ornstein-Uhlenbeck process starting from $\mu$ is $L$-smooth.
\end{itemize}
We show that under the assumptions (1*) and (2), the query complexity of sampling from $\mu$ can be reduced to $\!{poly}(L,d)\cdot \tp{\frac{Ld+M}{\eps^2}}^{\+O(L+1)}$, which is polynomial in $d$ and $\frac{1}{\eps}$ when $L=\+O(1)$ and $M=\!{poly}(d)$. This improves the algorithm with quasi-polynomial query complexity developed by Huang, Zou, Dong, Ma and Zhang (COLT, 2024). Our results imply that the seemingly moderate strengthening  of the smoothness condition (1) to (1*) leads to an exponential separation in the query complexity for sampling.

Furthermore, we show that together with the assumption (1*) and the stronger moment assumption that $\norm{X}$ is $\lambda$-sub-Gaussian for $X\sim\mu$, the \Poincare constant of $\mu$ is at most $\+O\tp{\lambda}^{2(L+1)}$. We can also establish a modified log-Sobolev inequality for $\mu$ under these conditions. As an application of our technique, we obtain a new estimate of the modified log-Sobolev constant for a specific class of mixtures of strongly log-concave distributions.

% \ctodo{Are the two new bounds correct?}

%\ctodo{Take closer look at the moment condition.}
\end{abstract}

\newpage
\setcounter{tocdepth}{2}
\tableofcontents

\newpage

\section{Introduction}
We study the problem of sampling from a probability distribution $\mu$ over $\bb R^d$ with density $\propto e^{-V}$, given query access to the value and gradient of its potential function $V:\bb R^d\to \bb R$. This is a fundamental computational task across various fields, including theoretical computer science, machine learning, and statistical physics. The problem of sampling has been studied extensively, resulting in numerous proposed algorithms (see~\cite{Che25} for a modern and comprehensive treatment).
% The problem of sampling has been studied extensively, and many algorithms have been proposed. Notable examples include Langevin dynamics and its variants (see \EG,~\cite{CB18, CCBJ18, RYMB19, Leh23, CEL+24, GTC25}), proximal samplers (see \EG, \cite{LST21,CCSW22,LC23,AC24,MW25}), diffusion model based algorithms (see \EG,~\cite{CCLLSZ23, CCLLLS23, GKL24, HZD24}). 
% In the case where $\mu$ is strongly-log-concave or satisfies isoperimetric inequalities, polynomially (with regard to $d$) many queries suffice to approximate $\mu$ (see \cite{}).

A long line of research established that these algorithms converge within polynomial-many queries (with regard to $d$) under various metrics when $\mu$ is log-concave (e.g., \cite{CB18, SL19,ZCL+23,CCSW22})%\ctodo{log-concave (not strongly) also have poly-query complexity?}\htodo{Oh yes! I remove the "strongly" here.}
, or satisfies good isoperimetric inequalities (e.g., \cite{VW19,MFH+23,CCSW22,CEL+24}). However, the problem becomes significantly more challenging beyond these well-behaved settings, and much less is known. For the lower bounds in the most general case, the results of~\cite{LRG18, HZ25} indicate that any sampling algorithm requires $\tp{\frac{LM}{\eps d}}^{\Omega(d)}$ number of queries in general for the class of distributions that satisfy only the following two minimal assumptions, where $\eps$ is the desired accuracy in total variation distance:

%\ctodo{Maybe assump 1 is smoothness and assump 2 is moment condition?}

\begin{assumption}\label{assump:smooth}
    The potential function $V$ is differentiable and $L$-smooth, \IE, for any $x,y\in \bb R^d$, $\|\grad V(x)-\grad V(y)\| \leq L \|x-y\|$.
\end{assumption}

\begin{assumption}\label{assump:moment}
    The second moment $\E[X\sim \mu]{\|X\|^2}\leq M$ for some $M<\infty$.
\end{assumption}

These distributions do not necessarily satisfy the log-concavity property or good isoperimetric inequalities, leading to a growing interest in understanding the extent to which efficient sampling algorithms exist under these minimal conditions.
%There has been growing interest in understanding the conditions under which efficient samplers can be designed for distributions that do not necessarily satisfy log-concavity or good isoperimetric inequalities. 
Some previous works show that efficient sampling is still possible when the target distributions possess specific structures, such as being a mixture of Gaussians or having a similar shape (\EG, see~\cite{LRG18}). Other works derive convergence bounds in terms of intricate parameters of the target distribution, such as the action of a curve (\cite{GTC25}). These works primarily reply on annealing or tempering methods.

Another parallel line of research shows that denoising diffusion probabilistic models (DDPMs) can sample from the target distribution $\mu$ with polynomial queries under very weak conditions, as long as the score functions along the Ornstein-Uhlenbeck process (OU process) can be estimated efficiently (e.g., \cite{CCLLSZ23, CCLLLS23, CLL23, GCC24}).
% which are constructed based on the time reversal of the Ornstein-Uhlenbeck process (OU process) $\set{X^{\OU}(t)}_{t\geq 0}$ starting from $\mu$. Let $\xi^{\OU}_t$ be the law of $X^{\OU}(t)$. It has been shown that, under a stronger smoothness condition, \Cref{assump:smoothplus}, DDPM-based algorithms can simulate the target distribution $\mu$ with a query complexity polynomial in $M,d,\sup_{t\geq 0} L_t^{\OU}$ and the accuracy $\eps$, assuming access to a sufficiently accurate estimate of the score functions along the process (e.g., \cite{CCLLSZ23, CCLLLS23, CLL23, GCC24}).

Following this line, the work of \cite{HDHMZ24,HZD24,HRT24} attempts to further provide a direct implementation of score function estimation without relying on neural networks. Let $\set{X^{\OU}(t)}_{t\geq 0}$ be the OU process starting from $\mu$ and $\xi^{\OU}_t$ be the law of $X^{\OU}(t)$. The work \cite{HZD24} designs a recursive procedure to estimate the score functions and shows that a total query complexity of at most $\exp\tp{\+O(\ol{L}^3)\cdot \!{polylog}(Ld+M)}$ can be achieved with $\ol L = \sup_{t\geq 0} L_t^{\OU}$. Crucially, this result relies on the regularity of $\set{\xi_t^{\OU}}_{t\geq 0}$, formulated as follows:
% We also remark that \Cref{assump:smoothplus} is standard in the study of diffusion-based sampling algorithms.
\begin{assumption}\label{assump:smoothplus}
    The potential function $-\log \xi^{\OU}_t$ is twice differentiable and $L^{\OU}_t$-smooth for any $t\in [0,\infty)$.
\end{assumption}

While \Cref{assump:smoothplus} is a prevalent technical condition in the analysis of diffusion-based algorithms (e.g., \cite{CCLLSZ23, CCLLLS23, CLL23, LLT23, GCC24}), its precise nature and implications remain a subject of discussion. On one hand, some studies refer to it as a standard or relaxible assumption (e.g., \cite{CDD23,HZD24}). On the other hand, several works have noted that this condition is non-trivial and difficult to verify, potentially excluding important classes of distributions (e.g., \cite{CLL23, BBDD24}).

Furthermore, it is worth noting that the dependence of the bound in \cite{HZD24} on the dimension $d$ is quasi-polynomial, whereas the established lower bound for general distributions is exponential (\cite{LRG18, HZ25}). The diverging interpretations of \Cref{assump:smoothplus} with this stark contrast in query complexity, naturally raises the following question:
\begin{quote}
	\emph{What is the role of \Cref{assump:smoothplus} in efficient sampling? In particular, could it enable a sampler with polynomial query complexity?}	
\end{quote}

In this work, we answer the above question by proposing a variant of the restricted Gaussian dynamics to sample from $\mu$ under \Cref{assump:smoothplus} with a query complexity polynomial in $d$. Furthermore, we can directly establish a \Poincare inequality and a modified log-Sobolev inequality when the target distribution satisfies \Cref{assump:smoothplus} and stronger moment bounds. Our results demonstrate that \Cref{assump:smoothplus} indeed plays an important role in enabling efficient sampling. It implicitly captures certain favorable structural properties of the distribution that \Cref{assump:smooth} fails to reflect. Our results are summarized in \Cref{sec:main-results}.

Technique-wise, our proofs imply that the recursive paradigm designed in~\cite{HZD24} is unnecessary. Our analysis of the restricted Gaussian dynamics differs from previous analyses for well-behaved distributions (\EG, \cite{LST21,CCSW22,LC23,MW25}) and is based on the recently developed path-wise technique for analyzing high-dimensional distributions (see~\cite{Eld22} for an introduction). The main ingredients of our analysis are the concatenation of localization schemes and the approximate conservation of variance and entropy for the stochastic localization processes introduced and popularized in~\cite{CE25}. Similar techniques or ideas have been used in sampling from discrete distributions, such as the Gibbs distribution of the Ising model and the hardcore model (see \EG,~\cite{CE25,CCTY25}). We will present an overview of our techniques in \Cref{sec:overview}.
% \ctodo{some refs for RGD.}

%We further remark that our analysis effectively relies only on the one-sided semi-log-concave condition of \Cref{assump:smoothplus}, i.e. $-\grad^2 \log \xi_t^{\!{OU}}\mge -L_t^{\!{OU}} \Id$. In comparison, previous studies such as \cite{HDHMZ24,HZD24} necessitate both sides of \Cref{assump:smoothplus} to bound their discretization errors.
%\htodo{Not sure if we should highlight this. Although they use two-sided bounds of Assump 3 for discretization, this can be relaxed to just initial smoothness. They mainly need Assump 3 for the recursion algorithm, but that also only requires the semi-log-concave property.}

\subsection{Main results}\label{sec:main-results}

Our first main result states that any distribution satisfying \Cref{assump:moment} and \ref{assump:smoothplus} can be approximated with an error of $\eps>0$ in total variation distance using polynomially many queries.

\begin{theorem}[A simplified version of \Cref{thm:main-ub2}]\label{thm:main-ub}
    There exists an algorithm which, for any target distribution $\mu$ satisfying \Cref{assump:moment} and \ref{assump:smoothplus}, outputs a sample from a distribution $\tilde \mu$ such that $\DTV(\tilde{\mu},\mu)\leq \eps$, with expected number of queries bounded by
    \[
        N=\wt{\+O}\tp{\ol{L}d\cdot (V(0)-\min V + d^2)}\cdot \tp{\frac{\ol{L}d+M}{\eps^2}}^{\+O(\ol L+1)},
    \]
    where $\ol L= \sup_{t\geq 0}L_t^{\OU}$.\footnote{Here the notation $\wt{\+O}$ subsumes some negligible logarithmic terms with regard to $\ol L,d,M, \eps^{-1}$ and the locations of particles appearing in the implementation of RGO. Further discussion on the dependence of the particles in RGO can be found in \Cref{sec:RGO} and the discussions in Appendix A of \cite{LST21}.}
    % \ctodo{$L, \ol L, L_t^{\!{OU}}?$.}
\end{theorem}
% \ctodo{weiweide explain $\wt O()$ here. Maybe we can simply say that the $\wt O$ hides some negligible terms appearing in the implementation of the RGO. See [LST] for details..(shuai guo)}
% \htodo{Here I simplify \Cref{thm:main-ub2} by using $\ol L$ to substitute $L$.}

Without loss of generality, we can regard $V(0)-\min V = \!{poly}(M,\ol{L},d)$. When $\ol L$ is bounded, the query complexity given in \Cref{thm:main-ub} is polynomial in $d$. This improves the $\exp\tp{\+O(\ol{L}^3)\cdot \!{polylog}(\ol{L}d+M)}$ bound in \cite{HZD24} by reducing the query complexity from quasi-polynomial to polynomial, which solves an open problem in \cite{HZ25}. We further remark that our analysis readily extends to KL divergence. By tracking entropy rather than variance in the proof of \Cref{thm:main-ub}, one can establish KL convergence with a comparable query complexity.

Although \Cref{thm:main-ub} only bounds the expected query complexity, we can obtain a worst-case query complexity bound of at most $\frac{N}{\eps}$ by imposing a hard cutoff after the $\frac{N}{\eps}$-th query. By a simple application of the Markov inequality, we can demonstrate that the sampler still maintains its accuracy guarantee. Comparing this bound with the $\tp{\frac{LM}{d\eps}}^{\Omega(d)}$ lower bound in \cite{HZ25} --- achieving under \Cref{assump:moment} and \ref{assump:smooth} --- the exponential gap reveals that the additional assumption, \Cref{assump:smoothplus}, which frequently appears in diffusion-based works, is indeed crucial for enabling efficient sampling.

% Our result reveals that the \Cref{assump:smoothplus} plays a crucial role in enabling efficient sampling.

\bigskip
\Cref{thm:main-ub} shows that distributions satisfying \Cref{assump:moment} and \ref{assump:smoothplus} admit a polynomial-time sampling algorithm. Our next result demonstrates that further strengthening the moment condition to \Cref{assump:momentplus} allows us to directly establish a \Poincare inequality and a modified log-Sobolev inequality for the target distribution.

\begin{assumption}\label{assump:momentplus}
    With $X\sim \mu$, the distribution of $\|X\|$ is a $\lambda$-sub-Gaussian distribution for some constant $\lambda>0$, i.e., $\E[X\sim \mu]{e^{r^2 \|X\|^2}}\leq e^{r^2 \lambda^2}$ for any $r\in [-1/\lambda, 1/\lambda]$.
    % \ctodo{An equivalent form.}
\end{assumption}
Note that \Cref{assump:momentplus} is stronger than assuming $X\sim\mu$ to be $\lambda$-sub-Gaussian and weaker than assuming $X$ satisfies $T_1$-transportation inequality with constant $\lambda^2$ (see \EG,  \cite[Theorem 4.8]{Van16}).

\begin{theorem}[A simplified version of \Cref{thm:PI,thm:mLSI}]\label{thm:main-PI}
    If \Cref{assump:smoothplus} and \ref{assump:momentplus} hold, then the distribution satisfies a \Poincare inequality with constant 
    \[
        C^{\!{PI}}_{\mu}\leq \min_{s\in \left(0,\frac{\log 2}{4\lambda^2}\right]} \frac{2}{2-e^{4s\lambda^2}}\cdot \tp{\frac{s+1}{s}}^{\ol L +1},
    \] 
    and satisfies a modified log-Sobolev inequality with constant
    \[
        C^{\!{mLSI}}_{\mu}\le \min_{s_0\in \left(0,\frac 1{12\lambda^2}\right]} \tp{3+\frac{12s_0\lambda^2+1}{2-e^{4s_0\lambda^2}}}\cdot \tp{\frac{s+1}{s}}^{\ol L +1},
    \]
    where $\ol L= \sup_{t\geq 0}L_t^{\OU}$.
\end{theorem}
% \htodo{finish this thm}
% \htodo{Change the notation of $\gamma$}
% \Cref{assump:momentplus} can be viewed as a slightly stronger version of the sub-Gaussian condition. Specifically, when $d=1$, it is equivalent to sub-Gaussianity. We note that, in general, the parameter $\lambda$ may depend on the dimension $d$. 
% Furthermore, we remark that establishing the \Poincare inequality in \Cref{thm:main-PI} actually requires a weaker condition. It suffices for \Cref{assump:momentplus} to hold only for the specific function $f(x) = \norm{x}_2$. However, for simplicity of exposition, we state the results using the general \Cref{assump:momentplus}.
% \ctodo{An example demonstrating the usefulness of assump 3.}

It is worth noting that \Cref{assump:smoothplus} plays a crucial role in establishing \Cref{thm:main-PI}. If it were replaced by the weaker \Cref{assump:smooth}, the conclusion would no longer hold. For example, consider the distributions constructed in Section 3 of \cite{HZ25}, which satisfies \Cref{assump:momentplus} with parameter $\lambda = \+O\tp{\sqrt{M}}$ (we choose $\eps=\+O(1)$ in their setting). Suppose, for the sake of contradiction, that the \Poincare inequality in \Cref{thm:main-PI} still holds under only \Cref{assump:smooth}. Then it implies a polynomial-time sampler for those distributions, which contradicts the exponential lower bound in \cite{HZ25}. This indicates that compared to \Cref{assump:smooth}, \Cref{assump:smoothplus} implicitly captures some additional geometric properties of the distribution.

We also apply our method to typical multimodal distributions. Using analytical techniques similar to those in \Cref{thm:main-PI}, we explicitly compute the modified log-Sobolev constant for a specific class of mixture of strongly log-concave distributions. 
\begin{theorem}\label{thm:main-mix}
    Suppose $\rho$ is an $m$-strongly log-concave and $L$-log-smooth distribution and $\nu$ is supported inside a Euclidean ball with radius $R>0$. The modified log-Sobolev constant of mixture distribution $\mu=\nu*\rho$ satisfies
    \[
      C^{\!{mLSI}}_{\mu}\leq \frac 1{2m}\cdot e^{LR^2}.
    \]
    % Suppose $\nu$ is supported on $\+B_R$ and $\mu = \nu * \+N(0, \Sigma)$. We have $C^{\!{PI}}_{\mu}\leq \norm{\Sigma}_{\!{op}}\cdot e^{\lambda_{\min}(\Sigma)^{-1}\cdot R^2}$.
\end{theorem}

As a direct corollary, when $\rho$ is the Gaussian distribution $\+N(0,\Sigma)$, \Cref{thm:main-mix} recovers the result in \cite{MS23}, which is the best known bound on the modified log-Sobolev constant for this class of mixture distributions. Specifically, letting $\lambda_{\min}(\Sigma)$ be the minimum eigenvalue of $\Sigma$, we have $C^{\!{mLSI}}_\mu\leq \frac 12\norm{\Sigma}_{\!{op}}\cdot e^{\lambda_{\min}(\Sigma)^{-1}\cdot R^2}$.
% Since the \Poincare constant is invariant under translation, our theorem subsumes the cases where $\mu$ is a mixture of Gaussians with the same covariance and close centers. Prior to this work, the best known bound on the \Poincare constant for this class of mixture distributions was given by \cite{JNFP18}. By combining their Theorem 1.2 with a suitable scaling argument (like \Cref{lem:scalePI}), one can derive a bound of $\norm{\Sigma}_{\!{op}}\cdot e^{4\lambda_{\min}(\Sigma)^{-1}\cdot R^2}$. Our result in \Cref{thm:main-mix} improves upon this by reducing the constant in the exponent. 

\smallskip
Moreover, the concatenation argument (see \Cref{sec:overview} for a brief introduction) used in our proof of \Cref{thm:main-PI,thm:main-mix} is of independent interest, and might be extended to establish functional inequalities for certain distributions beyond \Cref{assump:momentplus}, as long as the evolution of the smoothness along the OU process can be effectively analyzed.
% \htodo{Check whether the previous best result is the one in \cite{JNFP18}.}

\subsection{Technical overview}\label{sec:overview}
Our algorithm is a variant of the restricted Gaussian dynamics, which is also known as the proximal sampler (see \EG \cite{LST21}). The key step in the convergence analysis is to bound the \Poincare constant of the restricted Gaussian dynamics. This not only implies the rapid mixing of the algorithm, but also allows us to use a concatenation argument to establish the functional inequalities in \Cref{thm:main-PI,thm:main-mix}. 

\paragraph{\Cref{assump:smoothplus} and the covariance bound of the stochastic localization process} Before delving into the details, we first clarify the relationship between the OU process and stochastic localization process, to identify the condition that \Cref{assump:smoothplus} corresponds to in the stochastic localization framework.

Consider the stochastic differential equation of the OU process
\[
    \dd X^{\OU}(t) = - X^{\OU}(t) \dd t + \sqrt{2} \dd B(t),\ X^{\OU}(0)\sim \mu, 
\]
and the scheme that induces the stochastic localization process
\[
    X(s) = s\cdot X + B(s),\ X\sim \mu.
\]
Intuitively, the random variables $X^{\OU}(t)$ and $X(s)$ are related through a scaling transformation. Indeed, via direct calculations, we can prove that $X^{\!{OU}}(t)$, and $\sqrt{\frac{1}{s(1+s)}}\cdot X(s)$ with $s = \frac{e^{-2t}}{1-e^{-2t}}$, have the same distribution. Let $\nu_s(\cdot\,|\,z)$ be the conditional law of $X$ given $X(s)=z$, and when the information of $X(s)$ is clear, we may write it as $\nu_s$ for brevity. Building on this correspondence between the two processes, we can examine \Cref{assump:smoothplus} from the perspective of the stochastic localization process. To be specific, it translates into the following condition on the covariance of $\+T_v \nu_s(\cdot\,|\,z)$, where $\+T_v \nu_s(\cdot\,|\,z)$ is the exponentially tilted distribution with $\+T_v \nu_s(x\,|\,z) \propto e^{\inner{v}{x}}\cdot \nu_s(x\,|\,z)$.

\begin{condition}\label{cond:cov}
    For any $s\in \left[0, \infty \right)$, any $v,z\in \bb R^d$, 
    \[
       \tp{ \frac{1}{s}-\frac{ L_s}{s(1+s)}}\cdot \Id \preceq  \cov{\+T_v\nu_s(\cdot\,|\,z)} \preceq \tp{\frac{1}{s} + \frac{ L_s }{s(1+s)}}\cdot \Id,
    \]
    where $L_s=L^{\OU}_{t}$ with $t=\log\sqrt{\frac{s+1}{s}}$.
\end{condition}

\paragraph{The challenges to prove \Cref{thm:main-ub,thm:main-PI} under \Cref{cond:cov}} 
To prove \Cref{thm:main-ub}, our target is to bound the \Poincare constant of the restricted Gaussian dynamics. This is a discrete-time Markov chain $\set{Y_k}_{k\geq 0}$ induced by the stochastic localization process. For a fixed $T>0$, the restricted Gaussian dynamics with transition kernel $\*P^{(T)}$ executes as follows at each iteration  $k$:
\begin{itemize}
    \item draw $\hat Y_k\sim \+N\tp{T\cdot Y_{k-1}, T\cdot \Id}$;
    \item then draw $Y_k\sim \nu_T\tp{\cdot\,|\,\hat Y_k}$.
\end{itemize}
Let $\xi_s$ be the law of $X(s)$. From Proposition~\ref{prop:PIofSL}, the \Poincare constant of this chain is
\[
    C^{\!{PI}}_{\mu}\tp{\*P^{(T)}} = \sup_{f\colon \bb R^d\to \bb R} \frac{\Var[\mu]{f}}{\E[X(s)\sim \xi_T]{\Var[\nu_T]{f}}}.
\]
On the other hand, to prove \Cref{thm:main-PI},\footnote{For the sake of brevity, we only illustrate the proof strategy for the \Poincare inequality in \Cref{thm:main-PI} in this section. The modified log-Sobolev inequality can be derived via a parallel framework. See \Cref{sec:concatenation} for details.} we aim to directly bound the \Poincare constant of the target distribution $\mu$, which satisfies
\begin{align}
    C^{\!{PI}}_{\mu} &= \sup_{f:\bb R^d \to \bb R} \frac{\Var[\mu]{f}}{\E[\mu]{\|\grad f\|^2}} = \sup_{f:\bb R^d \to \bb R} \frac{\E[\xi_T]{\Var[\nu_T]{f}}}{\E[\mu]{\|\grad f\|^2}}\cdot \frac{\Var[\mu]{f}}{\E[\xi_T]{\Var[\nu_T]{f}}} \notag \\
    &\leq C^{\!{PI}}_{\mu}\tp{\*P^{(T)}}\cdot \sup_{f:\bb R^d \to \bb R} \frac{\E[\xi_T]{\Var[\nu_T]{f}}}{\E[\xi_T]{\E[\nu_T]{\|\grad f\|^2}}}. \label{eq:overview-1}
\end{align}
From the result on the approximate conservation of variance (see \Cref{thm:conservation}), given \Cref{cond:cov}, $C^{\!{PI}}_{\mu}\tp{\*P^{(T)}} \leq \exp\set{\int_0^T \tp{\frac{1}{s} + \frac{ L_s }{s(1+s)}} \dd s}$. However, two key technical challenges remain:
\begin{itemize}
    \item the integral $\int_0^T \tp{\frac{1}{s} + \frac{ L_s }{s(1+s)}} \dd s$ diverges even for small $T$;
    \item one needs to sample from $\nu_T\tp{\cdot\,|\,\hat Y_k}$ in the algorithm efficiently (\Cref{thm:main-ub}), or bound the \Poincare constant for $\nu_T$ (\Cref{thm:main-PI}).%, and to bound the ratio $\sup_{f:\bb R^d \to \bb R} \frac{\E[\xi_T]{\Var[\nu_T]{f}}}{\E[\xi_T]{\E[\nu_T]{\|\grad f\|^2}}}$ in \cref{eq:overview-1} need to be handled carefully.
\end{itemize}

\paragraph{Late initialization and three-phase concatenation}
To address the two issues discussed above, we analyze the evolvement of the variance along the stochastic processes in three phases. Note that in the stochastic localization process $\set{\nu_s}_{s\geq 0}$, 
\[
    \nu_s(x)\propto \mu(x)\cdot \exp\set{-\frac{\norm{X(s) - sx}^2}{2s}}.
\] 
Then we have the following three observations:
\begin{enumerate}
    \item $\nu_{s_0}\approx \nu_0$ and $\xi_{s_0} \approx \+N(0,s_0(s_0+1)\Id)$ for small enough $s_0$;
    \item $\int_{s_0}^T \tp{\frac{1}{s} + \frac{ L_s }{s(1+s)}} \dd s$ is bounded for any $0<s_0\leq T$;
    \item if $\mu$ satisfies \Cref{assump:smooth}, then $\nu_T$ is $(T-L)$-strongly log-concave for $T>L$.
\end{enumerate}

\begin{figure}[H]
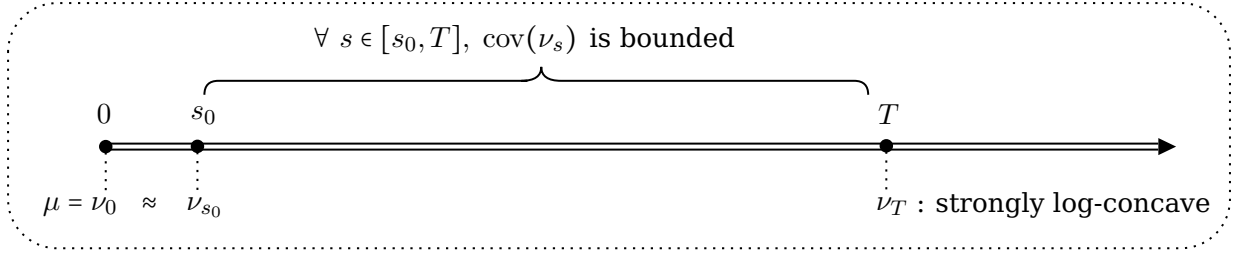

	\centering
    \setlength{\abovecaptionskip}{0pt}
    \ExecuteMetaData[figure.tex]{SLfigure}
  \caption{The stochastic localization process of $\mu$}
  \label{fig:SL}
\end{figure}
As shown in \Cref{fig:SL}, we divide the stochastic localization process into three phases and analyze $\set{\nu_s}_{s\in [0,s_0)}$, $\set{\nu_s}_{s\in [s_0,T)}$ and $\nu_T$ respectively.

We first explain how to resolve the challenges in the proof of \Cref{thm:main-ub}. For the sake of clarity, some notation is simplified and full technical details appear in the formal proofs in \Cref{sec:algo}. Note that to sample from $\mu$, it suffices to draw $X(s_0)\sim \xi_{s_0}$ and then sample from $\nu_{s_0}(\cdot\,|\,X(s_0))$. By the first observation, $\xi_{s_0}$ can be approximated by a Gaussian distribution $\+N(0,s_0(s_0+1)\cdot\Id)$. Thus, the target is reduced to simulate $\nu_{s_0}$. We call this a \emph{late initialization}. Then the second observation implies that, although bounding the \Poincare constant of the restricted Gaussian dynamics with respect to $\mu$ might be difficult, it is possible to control the \Poincare constant of the dynamics associated with the new target $\nu_{s_0}$ as the integral $\int_{s_0}^T \tp{\frac{1}{s} + \frac{ L_s }{s(1+s)}} \dd s$ is bounded. Of course, implementing this restricted Gaussian dynamics with respect to $\nu_{s_0}$ again involves the issue of sampling from $\nu_T$. By the third observation, with large $T$, $\nu_T$ becomes strongly log-concave, allowing efficient sampling using algorithms like rejection sampling. Combining all these, we can prove the rapid mixing of the restricted Gaussian dynamics with this late initialization scheme.

For the difficulty in the proof of \Cref{thm:main-PI}, we handle it in a similar way. We view the \Poincare constant of $\mu$ as the result of a three-segment concatenation:
\begin{align*}
    C^{\!{PI}}_{\mu}&= \sup_{f:\bb R^d \to \bb R}  \underbrace{\frac{\Var[\mu]{f}}{\E[\xi_{s_0}]{\Var[\nu_{s_0}]{f}}}}_{(\Rmnum{1})} \cdot \underbrace{\frac{\E[\xi_{s_0}]{\Var[\nu_{s_0}]{f}}}{\E[\xi_T]{\Var[\nu_T]{f}}}}_{(\Rmnum2)} \cdot \underbrace{\frac{\E[\xi_T]{\Var[\nu_T]{f}}}{\E[\mu]{\|\grad f\|^2}}}_{(\Rmnum3)}.
\end{align*}
For part $(\Rmnum{1})$, as $s_0$ is small, we can derive an explicit bound of this ratio under \Cref{assump:momentplus}. For part $(\Rmnum2)$, the second observation ensures that the integral $\int_{s_0}^T \tp{\frac{1}{s} + \frac{ L_s }{s(1+s)}} \dd s$ is bounded, which allows us to control this term accordingly. For part $(\Rmnum{3})$, the third observation implies that $\nu_T$ is $(T-L)$-strongly log-concave. As a result, its \Poincare constant is bounded almost surely by $\frac{1}{T-L}$ (see \EG,~\cite[Section 4.8]{BGL14}), which provides a direct upper bound for $(\Rmnum{3})$. Finally, we can derive the desired bound on the \Poincare constant $C^{\!{PI}}_{\mu}$ by concatenate the three parts.%\ctodo{Reference for (III)?}

We remark that the bound in part $(\Rmnum{1})$ is trivial in some other localization schemes such as the one associated to the field dynamics for sampling from hardcore model as well as the F\"ollmer process for sampling from the Ising model (\cite{CCTY25}) since an $\+O(\sqrt{d})\cdot \Id$ upper bound for the covariances in that part always holds (see also~\cite{CJ25}). However, in our case, a uniform upper bound (over every $y$) for $\cov{\nu_s(\cdot\,|\, y)}$ does not exist and therefore our special treatment of the first phase is necessary.

\smallskip
Finally, we prove \Cref{thm:main-mix} by applying the same concatenation argument to the mixture of strongly log-concave distributions. In order to obtain an upper bound for the covariance of distributions along the stochastic localization process, we again relate the quantity to smoothness of the potential function along the OU process, which can be bounded under the conditions in \Cref{thm:main-mix}.
% and \emph{bootstrap} the bound to the variance of a tilted distribution, which, in turn, can be bounded since the distribution is supported on a bounded set.

\subsection{Organization}
We begin by introducing some preliminaries in \Cref{sec:prelim}. In \Cref{sec:SLvsOU}, we establish a connection between the stochastic localization process and the OU process, which allows us to derive \Cref{cond:cov} from \Cref{assump:smoothplus}. Building on this foundation, we present our main algorithm and its analysis in \Cref{sec:algo} to prove \Cref{thm:main-ub}. Finally, we prove the functional inequalities stated in \Cref{thm:main-PI} and \Cref{thm:main-mix} in  \Cref{sec:concatenation}.

% \subsection{Related work}

\section{Preliminaries}\label{sec:prelim}
\subsection{Notations}
Throughout this paper, all distributions are assumed to be absolutely continuous with respect to the Lebesgue measure. For simplicity, we slightly abuse the notation and use the same symbol to denote both the distribution and its density. 
For random variables $X\in\bb R^d$ and $Y\in\bb R^d$, we use $p_X$ to denote the distribution of $X$, $p_{X,Y}$ to denote the joint distribution and $p_{X|Y}(\cdot |y)$ to represent the conditional distribution of $X$ given $Y=y$. 
Given a distribution $\mu$, we define $\m{\mu} = \E[X\sim \mu]{X}$, $\var{\mu} = \Var[X\sim \mu]{X}$ (or $\cov{\mu}=\Cov[X\sim \mu]{X}$ in high-dimensional cases) to denote its expectation and variance (or covariance) respectively. 

For a measure $\mu$ over $\bb R^d$ and a vector $v\in \bb R^d$, define the exponentially tilted distribution $\+T_v \mu$ as
$$
    \forall x\in \bb R^d,\ \+T_v \mu(x) \propto e^{\inner{v}{x}} \mu(x).
$$

We use $\+N(u, \Sigma)$ to denote the multivariate Gaussian distribution with mean $u\in \bb R^d$ and covariance $\Sigma\in \bb R^{d\times d}$. The notation $\+N(x; u, \Sigma)$ represents the density of $\+N(u, \Sigma)$ at $x$. Denote the Poisson distribution with mean $\lambda$ as $\!{Pois}(\lambda)$.%\ctodo{Or $\phi_{u,\Sigma}(x)$?} 

In this paper, $\log$ refers to the natural logarithm with base $e$. For two distributions $\mu_1$ and $\mu_2$ over $\bb R^d$, their total variation distance is $\DTV(\mu_1,\mu_2) = \tfrac{1}{2}\cdot \int_{\bb R^d} \abs{\mu_1(x)-\mu_2(x)} \dd x$.  Assuming $\mu_1 \ll \mu_2$, the Kullback-–Leibler divergence (KL divergence), $\chi^2$ divergence and \Renyi divergence with parameter $q> 1$ are respectively defined as
%\ctodo{Mention the use of $\mu(x)$ as density here.}
%\htodo{This is mentioned at the beginning of 2.1.}
\begin{itemize}
    \item $\DKL{\mu_1}{\mu_2} = \int_{\bb R^d} \mu_1(x)\log \frac{\mu_1(x)}{\mu_2(x)} \dd x$;
    \item $\chi^2(\mu_1\,\|\,\mu_2) = \int_{\bb R^d} \frac{\tp{\mu_1(x)-\mu_2(x)}^2}{\mu_2(x)} \dd x$;
    \item $\+R_q (\mu_1\,\|\,\mu_2) = \frac{1}{q-1}\log \tp{\int_{\bb R^d} \tp{\frac{\mu_1(x)}{\mu_2(x)}}^{q} \mu_2(x) \dd x}$.
\end{itemize}
Specifically, $\+R_\infty (\mu_1\|\mu_2) = \log\tp{\sup_{x\in \bb R^d} \frac{\mu_1(x)}{\mu_2(x)}}$.
By definition, $\+R_2 (\mu_1\|\mu_2) = \log\tp{1+\chi^2(\mu_1\|\mu_2)}$. For any $p\geq q>1$, $\+R_q (\mu_1\|\mu_2) \leq \+R_p (\mu_1\|\mu_2)$.

For a matrix $A\in \bb R^{d\times d}$, its operator norm $\|A\|_{\!{op}}$ is defined as $\|A\|_{\!{op}} = \sup_{y\in \bb R^d\atop \|y\|=1} \|Ay\|_2$.

%\ctodo{Define mgf and log-mgf.}

\subsection{The Markov chain, \Poincare inequality and modified log-Sobolev inequality}
 Consider a Markov chain with state space $\bb R^d$, transition kernel $\*P$ and stationary distribution $\mu$. 
%  , we say it is reversible with regard to $\mu$ if for all $x,y\in \bb R^d$,
%  \ctodo{I guess the detailed balance condition should be written in terms of density function (instead of transition kernel).}
%  \[
%     \mu(x)\cdot \*P(x,\dd y) = \mu(y)\cdot \*P(y,\dd x).
%  \]
 In this work, we only consider those distributions that admit strictly positive densities with respect to the Lebesgue measure. Therefore, we slightly abuse the notation by using $\*P(x,y)$ to denote the probability density function corresponding to the transition kernel $\*P(x, \dd y)$. Then we say $\*P$ is  reversible with regard to $\mu$ if for all $x,y\in \bb R^d$,
 \[
  \mu(x)\cdot \*P(x,y) = \mu(y)\cdot \*P(y,x).
\] 
 Define the Dirichlet form
\[
    \+E_{\*P}(f,g) = \frac{1}{2}\cdot \int_{\Omega\times \Omega} \tp{f(x)-f(y)}\tp{g(x)-g(y)} \mu(x)\*P(x,\dd y) \dd x .
\]
% \ctodo{We may use $C$ or $C_{\!{PI}}$ or $C^{\!{PI}}$ to denote \Poincare constant.}
We say $\mu$ satisfies a \Poincare inequality (PI) with regard to the Markov chain $\*P$ with constant $C$ if for all functions $f:\bb R^d \to \bb R$ with $\+E_{\*P}(f,f)\neq 0$,
\[
    \Var[\mu]{f}\leq C \cdot \+E_{\*P}(f,f).
\]
%\htodo{$f$ needs to be a smooth function or very regular function?}
%\ctodo{I don't think so. We alwasy choose $f(x) = \1{x\in A}$, which is not smooth.}

Similarly, $\mu$ is said to satisfy a modified log-Sobolev inequality (mLSI) with regard to $\*P$ with constant $C'$ if for all positive functions $f:\bb R^d \to \bb R_{> 0}$ with $\+E_{\*P}(f,\log f)\neq 0$,
\[
    \Ent[\mu]{f}\leq C' \cdot \+E_{\*P}(f,\log f),
\]
where the entropy
\[
    \Ent[\mu]{f} \defeq \E[\mu]{f\log f} - \E[\mu]{f}\cdot \log \E[\mu]{f}.
\]

Denote the \Poincare constant and modified log-Sobolev constant as $C^{\!{PI}}_{\mu}\tp{\*P} \defeq \sup_{f\colon \bb R^d\to \bb R} \frac{\Var[\mu]{f}}{\+E_{\*P}(f,f)}$ and $C^{\!{mLSI}}_{\mu}\tp{\*P} \defeq \sup_{f\colon \bb R^d\to \bb R_{> 0}} \frac{\Ent[\mu]{f}}{\+E_{\*P}(f,\log f)}$ respectively. Specifically, when the Markov chain is the Langevin dynamics with trajectory $\dd X(t) = - \grad V(X(t))\dd t + \sqrt{2}\dd B(t)$, it recovers the classical \Poincare inequality:
\begin{equation}
    \Var[\mu]{f}\leq C \cdot \E[\mu]{\norm{\grad f}^2} \label{eq:PI}
\end{equation}
and the classical modified log-Sobolev inequality:
\begin{equation}
    \Ent[\mu]{f}\leq C'\cdot \E[\mu]{\inner{\grad f}{\grad\log f}} = C' \cdot \E[\mu]{f^{-1}\norm{\grad f}^2}. \label{eq:mLSI}
\end{equation}

In this work, when the Markov chain is not explicitly specified, the terms \emph{\Poincare inequality} and \emph{modified log-Sobolev inequality} refer to \Cref{eq:PI,eq:mLSI}. Define $C^{\!{PI}}_{\mu} \defeq \sup_{f\colon \bb R^d\to \bb R} \frac{\Var[\mu]{f}}{\E[\mu]{\norm{\grad f}^2}}$ and $C^{\!{mLSI}}_{\mu} \defeq \sup_{f\colon \bb R^d\to \bb R_{>0}} \frac{\Ent[\mu]{f}}{\E[\mu]{f^{-1}\norm{\grad f}^2}}$. 
%\htodo{Do we need to define Langevin and prove its \Poincare inequality if they are only mentioned here?}
%\ctodo{I think we can simply write down the SDE of Langevin.}

\subsection{The stochastic localization process}\label{sec:prelim-SL}
%\ctodo{I vote for merging the SL section and the OU section.}
Let 
\begin{equation}
    X(s) = s\cdot X + B(s),\ X\sim \mu \label{eq:SL}
\end{equation}
with $\set{B(s)}_{s\geq 0}$ being a standard Brownian motion. Let $\xi_s$ be the law of $X(s)$ and $\nu_s(\cdot\,|\,y)$ be the conditional distribution of $X$ given $X(s)=y$. When the information of $y$ is clear, we will omit $y$ and use $\nu_s$ for simplicity. We can regard $\nu_s$ as a random distribution due to the randomness of $X(s)$. When $s=0$, $\nu_0=\mu$ and when $s=\infty$, $\nu_{\infty}$ is the Dirac distribution $\delta_{X(s)}$. The process $\set{\nu_s}_{s\geq 0}$ is the well-known stochastic localization process (SL process). 

For each $T\geq 0$, the SL process induces a natural discrete-time Markov chain $\set{Y_k}_{k\geq 0}$ with transition kernel $\*P^{(T)}$, named as the restricted Gaussian dynamics. 
% The natural discrete-time Markov chain $\set{Y_k}_{k\geq 0}$ induced by the SL process is called the restricted Gaussian dynamics. 
%\ctodo{Maybe mention the down-up view and say a little about the adjoint operator view (and how general this is).}
The transition kernel is defined as follows. In the $k$-th iteration,
\begin{itemize}
    \item (down-walk) draw $\hat Y_k\sim \+N\tp{T\cdot Y_{k-1}, T\cdot \Id}$;
    \item (up-walk) then draw $Y_k\sim \nu_T$ with $X(T)=\hat Y_k$.
\end{itemize}
The walk can be pleasantly interpreted as the \emph{down-up walk} along the SL process. The down-step walks from $\nu_\infty$ to $\nu_T$ by drawing $\hat Y_k\sim p_{X(T)\,|\,X(\infty)}(\cdot\,|\,Y_{k-1})$; and the up-step is simply its adjoint operator, or equivalently drawing $Y_k\sim p_{X(\infty)\,|\,X(T)}(\cdot\,|\,\hat Y_k)$.

From direct calculation, we know for any $x\in \bb R^d$ and any Borel set $A\subseteq \bb R^d$,
\begin{equation}
    \*P^{(T)}(x,A) = \E[X(T)\sim \xi_T]{\frac{\nu_T(x)\cdot \nu_T(A)}{\mu(x)}}. \label{eq:MC}
\end{equation}

% \htodo{$A \in \+B(\bb R^d)$ or $A\subseteq \bb R^d$?}
% \ctodo{Just say ``every Borel set $A\subseteq \bb R^d$''.}

The propositions and theorem below are well-known results for localization schemes.
% \ctodo{I think we should refer to the journal version of CE22 (Duke), with correct numbers of propositions and theorems.}
% \htodo{Done.}
\begin{proposition}[Fact 2.9 in \cite{CE25}]
The Markov chain with transition kernel $\*P^{(T)}$ is a reversible Markov chain with stationary measure $\mu$.
\end{proposition}

\begin{proposition}[Proposition 2.21 in \cite{CE25}]\label{prop:PIofSL}
    % The distribution $\mu$ satisfies 
    % \begin{itemize}
    %     \item $C^{\!{PI}}_{\mu}\tp{\*P^{(T)}} = \sup_{f\colon \bb R^d\to \bb R} \frac{\Var[\mu]{f}}{\E[X(t)\sim \xi_T]{\Var[\nu_T]{f}}}$ and
    %     \item $C^{\!{mLSI}}_{\mu}\tp{\*P^{(T)}} \leq  \sup_{f\colon \bb R^d\to \bb R_{>0}} \frac{\Ent[\mu]{f}}{\E[X(t)\sim \xi_T]{\Ent[\nu_T]{f}}}$.
    % \end{itemize}
    The distribution $\mu$ satisfies a \Poincare inequality with regard to the Markov chain $\*P^{(T)}$ with constant $C^{\!{PI}}_{\mu}\tp{\*P^{(T)}} = \sup_{f\colon \bb R^d\to \bb R} \frac{\Var[\mu]{f}}{\E[X(T)\sim \xi_T]{\Var[\nu_T]{f}}}$.
\end{proposition}
% \htodo{check the def of mLSI in \cite{CE22}}

The following theorem is a direct result in Sec 3.1.1 of \cite{CE25}. For the completeness of this paper, we also provide the proof in \Cref{sec:prelim-pf}.
\begin{theorem}[Consequence of approximate conservation of variance for SL process]\label{thm:conservation}
    Suppose for any $s\in[0,T]$, there exists some fixed $\theta_s>0$ such that $\|\cov{\nu_s}\|_{\!{op}} \leq \theta_s$ holds almost surely. Then
    \[
        C^{\!{PI}}_{\mu}\tp{\*P^{(T)}} \leq e^{\int_0^T \theta_s \d s}.
    \]
\end{theorem}

Similar results also hold for the mLSI constant. The proof of the following proposition is provided in \Cref{sec:prelim-pf}.
% \htodo{Similar results appear in \cite{CE25}. But the definition of mLSI there is not completely the same as our definition.}

\begin{proposition}\label{prop:mLSIofSL}
    % The distribution $\mu$ satisfies 
    % \begin{itemize}
    %     \item $C^{\!{PI}}_{\mu}\tp{\*P^{(T)}} = \sup_{f\colon \bb R^d\to \bb R} \frac{\Var[\mu]{f}}{\E[X(t)\sim \xi_T]{\Var[\nu_T]{f}}}$ and
    %     \item $C^{\!{mLSI}}_{\mu}\tp{\*P^{(T)}} \leq  \sup_{f\colon \bb R^d\to \bb R_{>0}} \frac{\Ent[\mu]{f}}{\E[X(t)\sim \xi_T]{\Ent[\nu_T]{f}}}$.
    % \end{itemize}
    The distribution $\mu$ satisfies a modified log-Sobolev inequality with regard to the Markov chain $\*P^{(T)}$ with constant $C^{\!{mLSI}}_{\mu}\tp{\*P^{(T)}} \leq  \sup_{f\colon \bb R^d\to \bb R_{>0}} \frac{\Ent[\mu]{f}}{\E[X(T)\sim \xi_T]{\Ent[\nu_T]{f}}}$.
\end{proposition}

The following result is a direct corollary of Proposition 3.19 and Lemma 3.20 in \cite{CE25}.
\begin{theorem}[Consequence of approximate conservation of entropy for SL process]\label{thm:mLSIconcervation}
    Suppose for any $s\in[0,T]$, there exists some fixed $\theta_s>0$ such that for any $v\in \bb R^d$, $\|\cov{\+T_v \nu_s}\|_{\!{op}} \leq \theta_s$ holds almost surely. Then
    \[
        \sup_{f\colon \bb R^d\to \bb R_{>0}} \frac{\Ent[\mu]{f}}{\E[X(T)\sim \xi_T]{\Ent[\nu_T]{f}}} \leq e^{\int_0^T \theta_s \d s}.
    \]
\end{theorem}

\subsection{The Ornstein-Uhlenbeck process}
The Ornstein-Uhlenbeck process (OU process) is a continuous-time stochastic process $\set{X^{\OU}(t)}_{t\geq 0}$ with the following trajectory:
\begin{equation}\label{eq:OU}
    \dd X^{\OU}(t) = - X^{\OU}(t) \dd t + \sqrt{2} \dd B(t).
\end{equation}
For any $t\geq 0$, $X^{\OU}(t)$ equals $e^{-t} X_{\OU}(0) + \sqrt{1-e^{-2t}}\zeta$ in distribution, where $\zeta\sim \+N(0,\Id)$ is independent with $X^{\OU}(0)$. Assume the process start at $X^{\OU}(0)\sim \mu$. Let $\xi^{\OU}_t$ be the law of $X^{\OU}(t)$ and $\nu^{\OU}_t(\cdot |y)$ be the conditional distribution of $X^{\OU}(0)$ given $X^{\OU}(t)=y$. When the information of $y$ is clear, we will omit $y$ and use $\nu^{\OU}_t$ for simplicity.

% The rapid convergence of the OU process is a well-established result (see, \EG,~\cite{BGL14,VW19,HZD24}). The proof of \Cref{prop:OUconvergence} is provided in \Cref{sec:prelim-pf}.
% % \htodo{This is Lemma C.3 in \cite{HZD24} actually. But in their statement, they assume the smoothness during the OU. So we'd better prove it again in our paper. }
% \begin{proposition}[Convergence of the OU process]\label{prop:OUconvergence}
%     Suppose $\mu$ satisfies Assumption~\ref{assump:smooth} and \ref{assump:moment}. Then
%     \[
%         \DKL{\xi^{\OU}_t}{\+N(0,\Id)} \leq e^{-\frac{t}{2}}\cdot (Ld+M).
%     \]
% \end{proposition}

\subsection{Restricted Gaussian oracle}\label{sec:prelim-RGO}
For some fixed $y\in \bb R^d$, $\sigma^2\in \bb R$, let $\mu_{y,\sigma^2}$ be the distribution with density $\propto \exp\set{-V(x) - \frac{\norm{x-y}^2}{2\sigma^2}}$.
The restricted Gaussian oracle (RGO) $\RGO(y,\sigma^2)$ takes as input a point $y\in \bb R^d$, a variance parameter $\sigma^2\in \bb R$. It outputs sample from $\mu_{y,\sigma^2}$. 
% In this work, we allow the RGO to return a possibly biased sample. The formal definition is as follows.
% \begin{definition}
%     The restricted Gaussian oracle $\RGO(y,\sigma^2,\eps)$ takes as input a point $y\in \bb R^d$, a variance parameter $\sigma^2\in \bb R$ and an accuracy level $\eps\in [0,1)$. It outputs a sample, whose distribution is within $\eps$ total variation distance of $\mu_{y,\sigma^2}$. 
% \end{definition}

A widely used method to implement the RGO is simply rejection sampling. In this work, we directly apply the rejection sampling algorithm and use the following result in \cite{LC23} as a black box. The details of the algorithm is provided in \Cref{sec:RGO}.

\begin{theorem}[A corollary of Propositions 3.2, 3.4 and D.4 in \cite{LC23}]\label{thm:rejection}
    Assume \Cref{assump:smooth} holds and $\sigma^2 \leq \frac{1}{Ld}$. For any $y\in \bb R^d$, there exists an algorithm that generates a sample from $\mu_{y,\sigma^2}$ with $\wt{\+O}(1)$ queries to $V$ and $\grad V$ in expectation.
\end{theorem}
% \htodo{The $\wt{\+O}$ omits log dependence of $L,d$ and the initial distance in the GD use by the rejection sampling algorithm.}
%\ctodo{I guess we should discuss what is hidden by $\wt{\+O}$.}

\section{The covariance of SL processes}\label{sec:SLvsOU}
%\ctodo{I think we do not need to say much about SL and OU. Just give their relation, and verify in appendix. Then all languages are in SL.}
%\ctodo{Remember to mention.}
We already mentioned that the stochastic localization (SL) process is simply a rescaling of the reversed Ornstein-Uhlenbeck (OU) process. We will use the SL scaling in most our proofs. In this section, we clarify the relationship between the SL scaling and OU scaling, and particularly, we translate the smoothness assumption \Cref{assump:smoothplus} to the covariance condition \Cref{cond:cov} in the context of the SL process.

%We first clarify the relationship between the two processes and demonstrate that \Cref{assump:smoothplus} translates to \Cref{cond:cov} in the context of the stochastic localization process. 

Recall the definition of the SL process $\set{\nu_s}_{s\geq 0}$ where $\nu_s$ is the conditional distribution of $X$ given $X(s)=s\cdot X + B(s)$, and the definition of the OU process $\set{X^{\!{OU}}(t)}_{t\geq 0}$ where $\dd X^{\!{OU}}(t) = - X^{\!{OU}}(t) \dd t + \sqrt{2} \dd B(t)$. 
% Let $p^{\OU}_{0|t}(\cdot |y)$ be the density of $X^{\!{OU}}(0)$ given $X^{\!{OU}}(t) = y$, and $p_{\infty |s}(\cdot |z)$ be the density of $X$ given $X(s)=z$. Similarly, denote the density of $X^{\!{OU}}(t)$ given $X^{\!{OU}}(0)=x$ and the density of $X(s)$ given $X=x$ by $p^{\OU}_{t|0}(\cdot | x)$ and $p_{s|\infty}(\cdot | x)$ respectively.

The following lemma gives the connection between the two processes and its proof is given in \Cref{sec:OUvsSL-pf}.

\begin{lemma}\label{lem:OUvsSL}
    Define $s = \frac{e^{-2t}}{1-e^{-2t}}$. With $X\sim \mu$ and $X^{\!{OU}}(0)\sim \mu$, 
    \begin{itemize}
        \item the distributions of $X^{\!{OU}}(t)$ and $\sqrt{\frac{1}{s(1+s)}}\cdot X(s)$ are the same for any $t>0$;
        % or equivalently $X(s) = \sqrt{s(1+s)}\cdot X^{\!{OU}}\tp{\frac{1}{2}\log\frac{1+s}{s}}$ for any $s>0$; 
        % \item for any $v\in \bb R^d$, the distribution $\nu_t^{\OU}(\cdot\, |\,y)$ is the same as the distribution $\+T_{v}\nu_s(\cdot\, |\,z)$ for $y=\frac{z+v}{\sqrt{s(1+s)}}$.
        \item the law of $X^{\!{OU}}(0)$ given $X^{\!{OU}}(t) = y$ is equal to the law of $X$ given $X(s) = \sqrt{s(1+s)}\cdot y$ for any $y\in \bb R^d$.
    \end{itemize}
\end{lemma}

%\ctodo{rewrite the proof in the appendix.}

The main result of this section is the following lemma stating that \Cref{assump:smoothplus} in the context of the OU process implies \Cref{cond:cov} in the context of the SL process.
%\htodo{\Cref{assump:smoothplus} implies \Cref{cond:cov}?}

\begin{lemma}\label{lem:cond1}
    \Cref{assump:smoothplus} implies \Cref{cond:cov}. More specifically, for an OU process $\tp{\xi^{\!{OU}}_t}_{t\ge 0}$ starting from $\xi^{\!{OU}}_0=\mu$, if $-\grad \log\xi^{\!{OU}}_t$ is $L^{\!{OU}}_t$-Lipschitz for any $t\in [0,\infty)$, then for an SL process $(\nu_s)_{s\ge 0}$ with $\nu_0=\mu$, it holds that for any $s\in \left[0, \infty \right)$, any $v,z\in \bb R^d$, 
    \[
        \frac{1+s - L_s}{s(1+s)}\cdot \Id \mle  \cov{\+T_v \nu_s(\cdot\,|\,z)} \mle \frac{1+s + L_s }{s(1+s)}\cdot \Id,
    \]
    where $L_s=L^{\OU}_{t}$ with $t=\log\sqrt{\frac{s+1}{s}}$.
\end{lemma}

The proof of Lemma~\ref{lem:cond1} is in \Cref{sec:proof-cond1}. Before that, we establish a few useful identities for SL processes.

\subsection{Some useful lemmas about score functions}

In this section, we derive a few useful identities for the density function $p_{X_s}(x)$ of the stochastic localization process 
\[
    \forall s\ge 0,\; X(s) = s\cdot X+B(s), \; X\sim\mu.
\]
We remark that similar calculations have been carried out in several places such as in the study of diffusion models (\EG,~\cite{CLL23}) in the context of OU process, and as the properties of \emph{logarithmic Laplace transform} (see \EG,~\cite{Eld22}).

\begin{lemma}\label{lem:tweedie}
    $
        \grad_y\log p_{X_s}(y) = \E[p_{X|X_s}(\cdot\,|\,y)]{X-y/s}.
    $
\end{lemma}
This identity is known as Tweedie's formula (see, e.g.~\cite{Efr11}). Its proof is given in \Cref{sec:OUvsSL-pf}. The first-order derivative of $\log p_{X_s}(y)$ gives the expectation of $p_{X|X_s}(\cdot | y)$. Alike, the second-order derivative reveals its covariance.

\begin{lemma}\label{lem:hessian-pt}
    $
        \grad^2_y \log p_{X_s}(y) = \cov{p_{X|X_s}(\cdot|y)} - \frac{\Id}{s}.
    $
\end{lemma}

\subsection{Proof of Lemma~\ref{lem:cond1}}\label{sec:proof-cond1}
    
By Lemma~\ref{lem:hessian-pt}, we know that 
\[
         \cov{\nu_s(\cdot |z)} = \grad_z^2 \log p_{X_s}(z) + \frac{\Id}{s}.
\]
From Lemma~\ref{lem:OUvsSL}, we know that 
\[
        p_{X_s}(z)  = \left(\frac{1}{\sqrt{s(1+s)}} \right)^d \cdot  p_{X_{t}^{\!{OU}}}\tp{\frac{z}{\sqrt{s(1+s)}}}.
\]
Thus
\[
         \grad_{z}^2 \log p_{X_s}(z) = \frac{1}{s(1+s)}\grad_{z}^2 \log p_{X_{t}^{\!{OU}}}\tp{\frac{z}{\sqrt{s(1+s)}}}.
\]
So 
\[
     -\frac{L_t^{\!{OU}}}{s(1+s)} \mle \grad_{z}^2 \log p_{X_s}(z) \mle \frac{L_t^{\!{OU}}}{s(1+s)}.
\]
Since $L_t^{\!{OU}} = L_s$, $\cov{\nu_s(\cdot |z)} = \grad_z^2 \log p_{X_s}(z) + \frac{\Id}{s}$. Thus, for any $z\in \bb R^d$,
\[
         \frac{1+s - L_s}{s(1+s)}\cdot \Id \mle  \cov{\nu_s(\cdot |z)} \mle \frac{1+s + L_s }{s(1+s)}\cdot \Id.
\]
Note that
\[
    \+T_v \nu_s(x\,|\,z) \propto \mu(x)\cdot \exp\set{-\frac{\|sx-z\|^2}{2s} + \inner{v}{x}} \propto\nu_s(x\,|\,z+v).
\]
This completes the proof of Lemma~\ref{lem:cond1}.

\section{The restricted Gaussian dynamics with late initialization}\label{sec:algo}

In this section, we prove \Cref{thm:main-ub} by introducing and analyzing the underlying algorithm. We first analyze an ideal continuous-time restricted Gaussian dynamics in \Cref{sec:ideal-RGD} and analyze its convergence. Our algorithm for proving \Cref{thm:main-ub} will be an approximate version of this ideal restricted Gaussian dynamics, described in \Cref{sec:late-RGD} and analyzed in \Cref{sec:late-RGD-analyze}.

\subsection{The convergence of the ideal restricted Gaussian dynamics}\label{sec:ideal-RGD}
We first prove the convergence of the restricted Gaussian dynamics introduced in \Cref{sec:prelim-SL}. For the ease of analysis, we will consider a continuous-time restricted Gaussian dynamics.

Recall that each iteration of the restricted Gaussian dynamics with transition kernel $\*P^{(T)}$ includes two steps:
\begin{itemize}
    \item (down-walk) draw $\hat Y_k\sim \+N\tp{T\cdot Y_{k-1}, T\cdot \Id}$;
    \item (up-walk) then draw $Y_k\sim \nu_T\tp{\cdot \,|\,\hat Y_k}$.
\end{itemize}
The first step, which is also called the down-walk step, is to sample from a Gaussian distribution and is straightforward to implement. For the second step, or the up-walk step, the target distribution $\nu_T\tp{\cdot\,|\,\hat Y_k}$ has the density function
\[
    \nu_T\tp{x \,|\,\hat Y_k} \propto \mu(x) \cdot \exp\set{-\frac{\norm{Tx-\hat Y_k}^2}{2T}} \propto \exp\set{-V(x) - \frac{\norm{Tx-\hat Y_k}^2}{2T}}.
\]
That is, the up-walk step requires the implementation of an exact RGO, which is $\RGO\tp{\frac{\hat Y_k}{T},\frac{1}{T}}$. 

We consider a continuous-time restricted Gaussian dynamics with heat kernel $\*H_t$ in a duration of $t$. The continuous-time process is constructed by involving a Poisson process with rate $1$. The heat kernel $\*H_t(x,\cdot)$ stands for the distribution of the particle at time $t$ by running the process starting from $x$. To be specific,
\begin{equation}
    \forall t\geq 0, \forall x,y\in \bb R^d,\,\*H_t(x,\dd y) = \sum_{k=0}^{\infty} \Pr{K_t=k}\cdot \tp{\*P^{(T)}}^k(x, \dd y), \label{eq:transition}
\end{equation}
where $K_t\sim \!{Pois}(t)$.

If the discrete-time chain $\*P^{(T)}$ has a bounded \Poincare constant $C^{\!{PI}}_{\mu}\tp{\*P^{(T)}}$, then the convergence of the corresponding continuous-time process with heat kernel $\*H_t$ is also guaranteed (see \Cref{sec:GD-pf} for details). 

The implementation of the continuous-time restricted Gaussian dynamics is given in \Cref{algo:GD}. To bound the worst case iteration complexity, we set a hard upper bound $K$ for the total number of iterations. We run the continuous-time Markov process for $t=\frac{K}{2}$ time and once the number of iterations exceeds $K$, we force the algorithm to terminate and output.

\begin{algorithm}[h]
	\caption{The continuous-time restricted Gaussian dynamics}
	\label{algo:GD}
        \Input{the times $T$, initial distribution $\mu_0$, total iteration number $K$}
	\begin{algorithmic}[1]
        \State Draw $Y_0\sim \mu_0$ 
	    \State Draw $K'\sim \!{Pois}(K/2)$, set $\+K\gets K\wedge K'$	
    \For{$k=1,2,\dots, \+K$}
		\State Sample $\hat Y_k \sim \+N(T\cdot Y_{k-1},\ T\cdot \Id)$\;
        \State Get $Y_{k}$ via $\RGO\tp{\frac{\hat Y_k}{T},\frac{1}{T}}$\;
		\EndFor
        \State Output $Y_{\+K}$\;
	\end{algorithmic}
\end{algorithm}

\Cref{thm:GD} below summarizes the rate of convergence of \Cref{algo:GD} and is proved in \Cref{sec:GD-pf}.
\begin{theorem}[Convergence of \Cref{algo:GD}]\label{thm:GD}
    Assume the \Poincare constant of $\mu$ with regard to the Markov chain $\*P^{(T)}$ is $C^{\!{PI}}_{\mu}\tp{\*P^{(T)}}$. Denote the output distribution of \Cref{algo:GD} as $p_{Y_{\+K}}$. Then 
    \[
        \!{TV}(p_{Y_{\+K}}\,\|\,\mu) \leq \tfrac12 \cdot e^{-K/{C^{\!{PI}}_{\mu}\tp{\*P^{(T)}}}} \cdot \sqrt{\chi^2(\mu_0\| \mu)} + 2e^{-{K}/{8}}
    \]
    for any $K\geq 6$. The iteration complexity of this algorithm is at most $K$.
\end{theorem}

\subsection{The practical algorithm with late initialization}\label{sec:late-RGD}
Recall that $\xi_s$ is the distribution of $X(s)$ in \cref{eq:SL} and $\nu_s$ is the law of $X$ given $X(s)$.
According to \Cref{thm:conservation}, if $\cov{\nu_s}$ is suitably bounded almost surely for all $s\in [0, T]$, then the \Poincare constant of the restricted Gaussian dynamics with transition kernel $\*P^{(T)}$ can be controlled. However, under \Cref{cond:cov} alone, we can only obtain the bound $\norm{\cov{\nu_s}}_{\!{op}}\leq \frac{1}{s} + \frac{L_s}{s(1+s)}$ and therefore the integration $\int_0^T \tp{\frac{1}{s} + \frac{L_s}{s(1+s)}} \dd s$ diverges, preventing us from directly concluding the convergence of the restricted Gaussian dynamics solely via conservation of variance. Observe that this divergence arises from the fact that $\frac{1}{s}$ becomes unbounded and is not locally integrable near zero.
% However, \Cref{cond:cov} does not guarantee such a bound for $s\in [0,s_0)$.
% \htodo{Explain the necessity of skipping.}

To address this, we need a late initialization scheme to skip the segment $[0,s_0)$ for some small $s_0$, whose value is to be determined later. The algorithm is given in \Cref{algo:modifiedGD}. Let us explain its main differences compared to \Cref{algo:GD}.

\begin{algorithm}[H]
	\caption{The restricted Gaussian dynamics with late initialization}
	\label{algo:modifiedGD}
        \Input{two time stamps $s_0$ and $T$}
	\begin{algorithmic}[1]
	    \State Draw $X_{s_0}\sim \+N\tp{0, s_0(1+s_0)\Id}$\;
        \State Draw $Y_0\sim \mu_0=\+N\tp{-\frac{\grad U_y(0)}{2(L+s_0)},\frac{I_d}{2(L+s_0)}}$\label{line:2}
        % \State Set $ K\gets \gamma\tp{\*Q^{(T)}} \cdot \log \frac{\sqrt{\chi^2\tp{\mu_0\|\nu_{s_0}(\cdot |X_{s_0})}}}{\eps}$
	    \State Draw $K'\sim \!{Pois}\tp{\frac{K_{X_{s_0}}}{2}}$, set $\+K\gets K_{X_{s_0}}\wedge K'$ 	
        \For{$k=1,2,\dots, \+K $}
		\State Sample $\hat Y_k \sim \+N(T\cdot Y_{k-1},\ T\cdot \Id)$\;
        \State Get $Y_{k}$ via $\RGO\tp{\frac{\hat Y_k + X_{s_0}}{T+s_0},\frac{1}{T+s_0}}$\;
		\EndFor
        \State Output $Y_{\+K}$\;
	\end{algorithmic}
\end{algorithm}

\paragraph{The late initialization} The first difference is the late initialization mentioned above. In \Cref{algo:modifiedGD}, we first draw $X_{s_0}\sim \xi_{s_0}$, and then use restricted Gaussian dynamics to generate a sample from the conditional distribution $\nu_{s_0}\tp{\cdot \,|\,X_{s_0}}$. Since we do not have direct access to $\xi_{s_0}$, we use $\+N\tp{0, s_0(1+s_0)\cdot\Id}$ to approximate $\xi_{s_0}$ instead as an approximation. In other words, we will consider the restricted Gaussian dynamics induced by the stochastic process 
\begin{equation}
    Z(s) = s\cdot Z + B(s),\ Z\sim \nu_{s_0}\tp{\cdot \,|\,X_{s_0}}. \label{eq:modifiedSL}
\end{equation}
For a fixed $X_{s_0}$, we let $\xi'_s$ and $\nu'_t(\cdot |y)$ denote the distribution of $Z(s)$ and the distribution of $Z$ given $Z(s)=y$. Denote $\*Q^{(T)}$ as the transition kernel of this restricted Gaussian dynamics. The \Poincare constant of this Markov chain is proved in the next lemma.

\begin{lemma}\label{lem:PI}
    Assume \Cref{cond:cov} holds. Given $X_{s_0}\in \bb R^d$, the restricted Gaussian dynamics induced by \cref{eq:modifiedSL} satisfies $C^{\!{PI}}_{\nu_{s_0}}\tp{\*Q^{(T)}} \leq \frac{s_0+T}{s_0}\cdot \exp\set{ \int_{s_0}^{T+s_0}  \frac{L_{s}}{s(1+s)} \d s}$.
\end{lemma}
\begin{proof}
    According to Lemma~\ref{lem:RGO}, under \Cref{cond:cov}, we have
    \begin{equation}
        \norm{\cov{\nu'_s}}_{\!{op}} \leq 
        \frac{1+s_0+s +L_{s+s_0}}{\tp{1+s_0+s}(s_0+s)} \label{eq:condition}
    \end{equation}
    almost surely. Then from \Cref{thm:conservation} and~\eqref{eq:condition}, 
    \begin{align*}
        C^{\!{PI}}_{\nu_{s_0}}\tp{\*Q^{(T)}} &\leq \exp\set{ \int_0^T \frac{1+s_0+s + L_{s+s_0}}{\tp{1+s_0+s}(s_0+s)} \d s} \\
        &= \exp\set{ \int_0^T \tp{ \frac{1}{s_0+s} + \frac{L_{s+s_0}}{(s_0+s)(1+s_0+s)}} \d s} \\
        &\leq \exp\set{\eval{\log(s_0+s)}_0^T + \int_0^T  \frac{L_{s+s_0}}{(s_0+s)(1+s_0+s)} \d s} \\
        &\leq \frac{s_0+T}{s_0}\cdot \exp\set{ \int_{s_0}^{T+s_0}  \frac{L_{s}}{s(1+s)} \d s}.
        % &= \exp\set{ \int_0^T \tp{ \frac{1}{s_0+t} + \frac{L_{t+s_0}}{s_0+t} - \frac{L_{t+s_0}}{1+s_0+t}} \dd t}  \\
        % &\leq \exp\set{(\ol L+1)\cdot \log(s_0+t)\big|_0^T - \ol L\cdot \log(1+s_0+t)\big|_0^T} \\
        % &\leq \frac{s_0+T}{s_0}\cdot \tp{\frac{1+s_0}{s_0}}^{\ol L}.
    \end{align*}
\end{proof}

\paragraph{Different implementation of the RGO} The second difference is the implementation of the RGO. Since the target distribution of the restricted Gaussian dynamics is $\nu_{s_0}\tp{\cdot\mid X_{s_0}}$ rather than $\mu$, we need to sample from a distribution with density 
\begin{equation}
    p_{Z|Z(T)}(x\,|\,\hat Y_k) \propto \nu_{s_0}\tp{x\,|\,X_{s_0}}\cdot \exp\set{-\frac{\norm{\hat Y_k - T x}^2}{2T}} \label{eq:conditional}
\end{equation}
in the up-walk step. The next result indicates that this can be achieved via $\RGO\tp{\frac{\hat Y_k + X_{s_0}}{T+s_0},\frac{1}{T+s_0}}$. 
% Note that here we also relax the requirement by permitting the RGO to have an error of at most $\eps_{\RGO}$ in total variation distance.
 
\begin{lemma}\label{lem:RGO}
    For any $x,x_{s_0},x_T\in \bb R^d$, $\nu_{T+s_0}(x\,|\,x_T + x_{s_0}) = \nu'_T(x\,|\,x_T)$.
\end{lemma}
\begin{proof}
    By the definition of conditional probability,
    \begin{align*}
        \nu_{T+s_0}(x\,|\,x_T + x_{s_0}) &= \frac{p_{X,X(T+s_0)}(x, x_T + x_{s_0})}{\xi_{T+s_0}(x_T + x_{s_0})}\\
        &\propto \mu(x)\cdot p_{X(T+s_0)|X}(x_T + x_{s_0}|x)\\
        &\propto \mu(x)\cdot \exp\set{- \frac{\norm{x_T + x_{s_0} - (T+s_0)x}^2}{2(T+s_0)}} \\
        &\propto \mu(x)\cdot \exp\set{ - \frac{(T+s_0)\|x\|^2}{2} + \inner{x_T + x_{s_0}}{x}}.
    \end{align*}
    Similarly,
    \begin{align*}
        \nu'_T\tp{x\mid x_T} &= \frac{p_{Z,Z(T)}\tp{x, x_T}}{\xi'_T(x_T)} \\
        &\propto p_{Z}(x)\cdot p_{Z(T)|Z}\tp{ x_T \mid x} \\
        \mr{definition of $Z$}&= \nu_{s_0}(x\,|\,x_{s_0}) \cdot p_{Z(T)|Z}\tp{ x_T \mid x} \\
        &\propto \mu(x)\cdot p_{ X(s_0) |X}(x_{s_0}\,|\,x) \cdot p_{Z(T)|Z}\tp{ x_T \mid x}\\
        &\propto \mu(x)\cdot \exp\set{-\frac{\norm{x_{s_0}-s_0\cdot x}^2}{2s_0}} \cdot \exp\set{-\frac{\norm{x_T - T\cdot x}^2}{2T}} \\
        &\propto \mu(x)\cdot \exp\set{ - \frac{(T+s_0)\|x\|^2}{2} + \inner{x_T+x_{s_0}}{x}}.
    \end{align*}
    Therefore, the two conditional distributions are the same.
\end{proof}

For any $y\in \bb R^d$, define the function $U_y(x) = V(x) + \frac{\norm{y - s_0\cdot x}^2}{2s_0}$. For the restricted Gaussian dynamics to converge, from \Cref{thm:GD}, it is enough to set $K$ as
\[
    \max\set{C^{\!{PI}}_{\mu}\tp{\*Q^{(T)}} \cdot \log \frac{2\sqrt{\chi^2\tp{\mu_0\|\nu_{s_0}(\cdot |X_{s_0})}}}{\eps}, 8\log \frac{8}{\eps}}
\]
From Lemmas~\ref{lem:PI} and \ref{lem:initialization}, the following $K_{X_{s_0}}$ is an upper bound of the above equation.
\begin{align}
    K_{X_{s_0}} &= \max\set{ 8\log \frac{8}{\eps}, \frac{s_0+T}{2s_0}\cdot \exp\set{ \int_{s_0}^{T+s_0}  \frac{L_{t}}{t(1+t)} \dd t}\cdot \left(\log(2e) + V(0) - \min V  \right.\right. \notag \\
    &\left.  \left. \qquad \qquad \qquad \qquad \qquad \qquad \qquad \qquad + \frac{d}{2}\log \tp{(L+s_0)\cdot  16e^2M}  + \frac{(2d+1)\norm{X_{s_0}}^2}{2s_0} + \log \frac{4}{\eps^2}\right)} \label{eq:complexity}
\end{align}

\subsection{The convergence of \Cref{algo:modifiedGD}}\label{sec:late-RGD-analyze}

In this section, we give the accuracy analysis and derive an expected query complexity bound of \Cref{algo:modifiedGD}. The following result is a corollary of \cref{eq:complexity} and \Cref{thm:GD}. 
\begin{corollary}\label{coro:modifiedGD}
    Assume \Cref{cond:cov} holds. In \Cref{algo:modifiedGD}, for any $y\in \bb R^d$, \[\DTV(p_{Y_{\+K}\,|\,X_{s_0}}(\cdot\,|\,y),\nu_{s_0}(\cdot\,|\,y)) \leq \frac{\eps}{2}.\]
\end{corollary}

The next lemma indicates that it is reasonable to use $\+N\tp{0, s_0(1+s_0)\cdot \Id}$ to approximate $\xi_{s_0}$. Its proof is given in \Cref{sec:algo-pf}.
\begin{lemma}\label{lem:OUconverge}
     Let $s_0 = \frac{\tp{\frac{\eps^2}{2(Ld+M)}}^4}{ 1- \tp{\frac{\eps^2}{2(Ld+M)}}^4}$ for some $\eps\in (0,1)$. Then $\DTV\tp{\xi_{s_0}, \+N\tp{0, s_0(1+s_0)\Id}}\leq \frac{\eps}{2}$.
\end{lemma}
% \begin{proof}
%     From \Cref{lem:OUvsSL}, $\xi_{s_0}(x)  = \tp{s_0(1+s_0)}^{-\frac{d}{2}}\cdot \xi^{\OU}_{t_0}(y)$ for $t_0=\log\sqrt{\frac{1+s_0}{s_0}}= 2\log\frac{Ld+M}{\eps_0}$ and $y= \frac{1}{\sqrt{s_0(1+s_0)}}x$.
%     By definition,
%     \begin{align*}
%         \DKL{\xi_{s_0}}{\+N\tp{0, s_0(1+s_0)\cdot\Id}} &= \int_{\bb R^d} \xi_{s_0}(x) \cdot \log\tp{\frac{\xi_{s_0}(x)}{\tp{2\pi\cdot s_0(1+s_0)}^{-\frac{d}{2}} \cdot e^{-\frac{\|x\|^2}{2s_0(1+s_0)}}}} \dd x \\
%         \mr{ $y= \frac{1}{\sqrt{s_0(1+s_0)}}x$ }&= \int_{\bb R^d} \xi^{\OU}_{t_0}(y)\cdot \log \frac{\xi^{\OU}_{t_0}(y)}{\tp{2\pi}^{-\frac{d}{2}} \cdot e^{-\frac{\|y\|^2}{2}}}\dd y \\
%         &= \DKL{\xi^{\OU}_{t_0}}{\+N\tp{0, \Id}}\\
%         \mr{\Cref{prop:OUconvergence}}&\leq \eps_0.
%     \end{align*}
%     The result then follows from the Pinsker's inequality.
% \end{proof}

Now we are ready to prove \Cref{thm:main-ub}.
\begin{theorem}\label{thm:main-ub2}
    With $s_0 = \frac{\tp{\frac{\eps^2}{2(Ld+M)}}^4}{ 1- \tp{\frac{\eps^2}{2(Ld+M)}}^4}$ and $T=2Ld$, the output of \Cref{algo:modifiedGD} satisfies $\DTV\tp{p_{Y_{\+K}}, \mu}\leq \eps$. The expected query complexity is bounded by
    \[
        \wt{\+O}\tp{Ld\cdot (V(0)-\min V + d^2)}\cdot\frac{1}{s_0}\cdot \exp\set{ \int_{s_0}^{T+s_0}  \frac{L_{s}}{s(1+s)} \d s}.
    \]
    In particular, letting $\ol L=\sup_{t\geq s_0} L_t$, the query complexity is at most
    \[
        \wt{\+O}\tp{Ld\cdot (V(0)-\min V + d^2)}\cdot \tp{\frac{Ld+M}{\eps^2}}^{\+O(\ol L+1)}.
    \]
\end{theorem}
% \ctodo{$\eps_0$ and $\eps$.}
\begin{proof}
    Note that
    \[
        p_{Y_{\+K}}(x) = \int_{\bb R^d} p_{X_{s_0}}(z) \cdot p_{Y_{\+K}|X_{s_0}}(x\,|\,z) \d z
    \]
    and
    \[
        \mu(x) = \int_{\bb R^d} \xi_{s_0}(z)\cdot \nu_{s_0}(x\,|\,z) \d z.
    \]
    Therefore, 
    \begin{align*}
        \DTV\tp{p_{Y_{\+K}}, \mu} 
        &= \tfrac{1}{2}\int_{\bb R^d} \abs{p_{Y_{\+K}}(x) - \mu(x)} \d x \\
        &= \tfrac{1}{2}\int_{\bb R^d} \abs{\int_{\bb R^d} p_{X_{s_0}}(z) \cdot p_{Y_{\+K}|X_{s_0}}(x\,|\,z) \d z  - \int_{\bb R^d} \xi_{s_0}(z)\cdot \nu_{s_0}(x\,|\,z) \d z} \d x \\
        \mr{triangle inequality}
        &\leq \tfrac{1}{2}\int_{\bb R^d} \abs{\int_{\bb R^d} p_{X_{s_0}}(z) \cdot p_{Y_{\+K}|X_{s_0}}(x\,|\,z) \d z  - \int_{\bb R^d} \xi_{s_0}(z)\cdot p_{Y_{\+K}|X_{s_0}}(x\,|\,z) \d z} \d x \\
        &\quad + \tfrac{1}{2}\int_{\bb R^d} \abs{\int_{\bb R^d} \xi_{s_0}(z) \cdot p_{Y_{\+K}|X_{s_0}}(x\,|\,z) \d z  - \int_{\bb R^d} \xi_{s_0}(z)\cdot \nu_{s_0}(x\,|\,z) \d z} \d x.
    \end{align*}
    By Jensen's inequality, we can bound the above by
    \begin{align*}
    &\tfrac{1}{2}\int_{\bb R^{d}\times \bb R^{d}} \abs{ p_{X_{s_0}}(z)  -  \xi_{s_0}(z) }\cdot p_{Y_{\+K}|X_{s_0}}(x\,|\,z)\d z \d x  + \tfrac{1}{2}\int_{\bb R^{d}\times \bb R^{d}} \abs{ p_{Y_{\+K}|X_{s_0}}(x\,|\,z)  -  \nu_{s_0}(x\,|\,z) }\cdot \xi_{s_0}(z)\d z \d x \\
        &\quad= \DTV\tp{\xi_{s_0}, \+N\tp{0, s_0(1+s_0)\cdot\Id}} + \int_{\bb R^{d}} \DTV(p_{Y_{\+K}|X_{s_0}}(\cdot\,|\,z),\nu_{s_0}(\cdot\,|\,z))\cdot \xi_{s_0}(z)\d z,
    \end{align*}
    which is at most $\eps$ by Corollary~\ref{coro:modifiedGD} and Lemma~\ref{lem:OUconverge}.

    With $T=2Ld$, according to \Cref{thm:rejection}, the implementation of $\RGO\tp{\frac{\hat Y_k + X_{s_0}}{T+s_0},\frac{1}{T+s_0}}$ needs $\wt{\+O}(1)$ queries in expectation. From \cref{eq:complexity}, the expected query complexity can be bounded by
    \begin{align*}
        &\phantom{{}={}}\E[X_{s_0}\sim \+N\tp{0, s_0(1+s_0)\cdot\Id}]{K_{X_{s_0}}}\\
        &= \wt{\+O}\tp{ \frac{s_0+T}{s_0}\cdot \exp\set{ \int_{s_0}^{T+s_0}  \frac{L_{s}}{s(1+s)} \d s} \cdot \tp{V(0)-\min V + d\log(LM) + d^2 + \log\frac{4}{\eps^2}} } \\
        &= \wt{\+O}\tp{Ld\cdot (V(0)-\min V + d^2)}\cdot \frac{1}{s_0}\cdot \exp\set{ \int_{s_0}^{T+s_0}  \frac{L_{s}}{s(1+s)} \d s}.
    \end{align*}
    This can be further bounded by
    \[
        \wt{\+O}\tp{Ld\cdot (V(0)-\min V + d^2)}\cdot \frac{1}{s_0}\cdot \tp{\frac{(T+s_0)(1+s_0)}{s_0(1+T+s_0)}}^{\ol L},
    \]
    and can be simplified to 
    \[
     \wt{\+O}\tp{Ld\cdot (V(0)-\min V + d^2)}\cdot \tp{\frac{Ld+M}{\eps^2}}^{\+O(\ol L+1)}.
     \]
\end{proof}

\section{The bound on the \Poincare constant and modified log-Sobolev constant}\label{sec:concatenation}
In this section, we introduce the concatenation of the \Poincare constant and modified log-Sobolev constant under \Cref{cond:cov} in \Cref{sec:concate-PI,sec:concate-mLSI}, and apply these methods to prove the two parts of \Cref{thm:main-PI} in \Cref{sec:sub-G-PI,sec:sub-G-mLSI} respectively. Finally, in \Cref{sec:mixofSLC}, we derive the modified log-Sobolev constant bound of the mixture of strongly log-concave distributions.
% \htodo{modify the intro here}

We first introduce a technical lemma. The next lemma is indeed a tower rule for conditional expectations with appropriate random variables along the SL trajectory. 
%shows a further relationship between the two stochastic localization schemes with target distributions $\mu$ and $\nu_{s}$. 
We leave the formal verification of Lemma~\ref{lem:beforeconcate} in \Cref{sec:concat-pf}.
\begin{lemma}\label{lem:beforeconcate}
    For any $0\leq s<T$, we have
    \begin{itemize}
        \item $\E[\xi_T]{\Var[\nu_T]{f}} = \E[\xi_{s}]{\E[\xi'_{T-s}]{\Var[\nu'_{T-s}]{f}}}$;
        \item $\E[\xi_T]{\Ent[\nu_T]{f}} = \E[\xi_{s}]{\E[\xi'_{T-s}]{\Ent[\nu'_{T-s}]{f}}}$.
    \end{itemize} 
\end{lemma}
%\ctodo{Change notation in appendix.}

\subsection{The concatenation of the \Poincare constant}\label{sec:concate-PI}
% \htodo{Explain the intuition of concatenation.}

Recall that $\*P^{(s_0)}$ is the transition kernel of the restricted Gaussian dynamics with target distribution $\mu$ at time $s_0$. 

\begin{theorem}\label{thm:concatenation}
    Assume \Cref{assump:smooth} and \Cref{cond:cov} hold and assume $C^{\!{PI}}_{\mu}\tp{\*P^{(s_0)}}<\infty$ for some $0\leq s_0\leq T=2L$. Then $\mu$ satisfies the \Poincare inequality with constant
    \[
        C^{\!{PI}}_{\mu}\leq C^{\!{PI}}_{\mu}\tp{\*P^{(s_0)}}\cdot \frac{2}{s_0}\cdot \exp\set{ \int_{s_0}^{2L}  \frac{L_{s}}{s(1+s)} \dd s}. 
    \]
\end{theorem}
\begin{proof}
    For any $T>s_0$, by the definition of the \Poincare inequality, we know
    \begin{align*}
        C^{\!{PI}}_{\mu} &= \sup_{f:\bb R^d \to \bb R} \frac{\Var[\mu]{f}}{\E[\mu]{\|\grad f\|^2}}  \\
        &= \sup_{f:\bb R^d \to \bb R} \frac{\E[\xi_T]{\Var[\nu_T]{f}}}{\E[\mu]{\|\grad f\|^2}}\cdot \frac{\E[\xi_{s_0}]{\Var[\nu_{s_0}]{f}}}{\E[\xi_T]{\Var[\nu_T]{f}}} \cdot \frac{\Var[\mu]{f}}{\E[\xi_{s_0}]{\Var[\nu_{s_0}]{f}}} \\
        &\leq \underbrace{\tp{\sup_{f:\bb R^d \to \bb R} \frac{\E[\xi_T]{\Var[\nu_T]{f}}}{\E[\mu]{\|\grad f\|^2}}}}_{\!{(A)}}\cdot \underbrace{\tp{\sup_{f:\bb R^d \to \bb R} \frac{\E[\xi_{s_0}]{\Var[\nu_{s_0}]{f}}}{\E[\xi_T]{\Var[\nu_T]{f}}}}}_{\!{(B)}}\cdot \underbrace{\tp{\sup_{f:\bb R^d \to \bb R} \frac{\Var[\mu]{f}}{\E[\xi_{s_0}]{\Var[\nu_{s_0}]{f}}}}}_{\!{(C)}}.
    \end{align*}
    Now we bound each term separately. For the term $\!{(C)}$, it is exactly $C^{\!{PI}}_{\mu}\tp{\*P^{(s_0)}}$ by Proposition~\ref{prop:PIofSL}.
    
    Note that when \Cref{assump:smooth} holds, by choosing $T=2L$, $\nu_T$ is $L$-strongly log-concave. By a well-known consequence of the Bakry-\'Emery criterion (see \EG, ~\cite[Section 4.8]{BGL14}), it satisfies \Poincare inequality with constant $\frac{1}{L}$ almost surely. Therefore,
    \[
        \!{(A)}= \sup_{f:\bb R^d \to \bb R} \frac{\E[\xi_T]{\Var[\nu_T]{f}}}{\E[\xi_T]{\E[\nu_T]{\|\grad f\|^2}}} \leq \frac{1}{L}.
    \] 
    From Lemma~\ref{lem:PI} and Lemma~\ref{lem:beforeconcate}, 
    \begin{align*}
        \!{(B)}= \sup_{f:\bb R^d \to \bb R} \frac{\E[\xi_{s_0}]{\Var[\nu_{s_0}]{f}}}{\E[\xi_{s_0}]{\E[\xi'_{T-s_0}]{\Var[\nu'_{T-s_0}]{f}}}} \leq \frac{T}{s_0}\cdot \exp\set{ \int_{s_0}^{T}  \frac{L_{s}}{s(1+s)} \dd s}.
    \end{align*}
    Combining all above, we have
    \[
        C^{\!{PI}}_{\mu}
        % \leq C^{\!{PI}}_{\mu}\tp{\*P^{(s_0)}}\cdot \frac{1}{L} \cdot\frac{2L}{s_0}\cdot \exp\set{ \int_{s_0}^{T}  \frac{L_{s}}{s(1+s)} \d s} 
        \leq C^{\!{PI}}_{\mu}\tp{\*P^{(s_0)}}\cdot \frac{2}{s_0}\cdot \exp\set{ \int_{s_0}^{2L}  \frac{L_{s}}{s(1+s)} \d s }. 
    \]
\end{proof}

% Note that when \Cref{assump:smooth} holds, choosing $T=2L$, $\nu_T$ is $L$-strongly log-concave. By a well-known consequence of the Bakry-\'Emery criterion (see \EG, ~\cite[Section 4.8]{BGL14}), it satisfies \Poincare inequality with constant $\frac{1}{L}$ almost surely. Then the following result is a direct corollary of \Cref{thm:concatenation}.
% \begin{corollary}
%     Assume \Cref{assump:smooth} and \Cref{cond:cov} hold and assume $C^{\!{PI}}_{\mu}\tp{\*P^{(s_0)}}<\infty$ for some $0\leq s_0\leq 2L$. Then $\mu$ satisfies the \Poincare inequality with constant
%     \[
%         C^{\!{PI}}_{\mu}\leq C^{\!{PI}}_{\mu}\tp{\*P^{(s_0)}}\cdot \frac{2}{s_0}\cdot \exp\set{ \int_{s_0}^{2L}  \frac{L_{s}}{s(1+s)} \dd s}. 
%     \]
% \end{corollary}

\subsection{The concatenation of the modified log-Sobolev constant}\label{sec:concate-mLSI}
Analogous to \Cref{sec:concate-PI}, this section analyzes the concatenation of the modified log-Sobolev constant.
\begin{theorem}\label{thm:concate-mLSI}
    Assume \Cref{assump:smooth} and \Cref{cond:cov} hold and assume $\sup_{f:\bb R^d \to \bb R_{>0}} \frac{\Ent[\mu]{f}}{\E[\xi_{s_0}]{\Ent[\nu_{s_0}]{f}}}<\infty$ for some $0\leq s_0\leq T=2L$. Then $\mu$ satisfies the modified log-Sobolev inequality with constant
    \[
        C^{\!{mLSI}}_{\mu}\leq \sup_{f:\bb R^d \to \bb R_{>0}} \frac{\Ent[\mu]{f}}{\E[\xi_{s_0}]{\Ent[\nu_{s_0}]{f}}}\cdot \frac{1}{s_0}\cdot \exp\set{ \int_{s_0}^{2L}  \frac{L_{s}}{s(1+s)} \dd s}. 
    \]
\end{theorem}
\begin{proof}
    For any $T>s_0$, by the definition of the modified log-Sobolev constant, we know
    \begin{align*}
        C^{\!{mLSI}}_{\mu} &= \sup_{f:\bb R^d \to \bb R_{>0}} \frac{\Ent[\mu]{f}}{\E[\mu]{f^{-1}\|\grad f\|^2}}  \\
        &= \sup_{f:\bb R^d \to \bb R_{>0}} \frac{\E[\xi_T]{\Ent[\nu_T]{f}}}{\E[\mu]{f^{-1}\|\grad f\|^2}}\cdot \frac{\E[\xi_{s_0}]{\Ent[\nu_{s_0}]{f}}}{\E[\xi_T]{\Ent[\nu_T]{f}}} \cdot \frac{\Ent[\mu]{f}}{\E[\xi_{s_0}]{\Ent[\nu_{s_0}]{f}}} \\
        &\leq \underbrace{\tp{\sup_{f:\bb R^d \to \bb R_{>0}} \frac{\E[\xi_T]{\Ent[\nu_T]{f}}}{\E[\mu]{f^{-1}\|\grad f\|^2}}}}_{\!{(A)}}\cdot \underbrace{\tp{\sup_{f:\bb R^d \to \bb R_{>0}} \frac{\E[\xi_{s_0}]{\Ent[\nu_{s_0}]{f}}}{\E[\xi_T]{\Ent[\nu_T]{f}}}}}_{\!{(B)}}\cdot \underbrace{\tp{\sup_{f:\bb R^d \to \bb R_{>0}} \frac{\Ent[\mu]{f}}{\E[\xi_{s_0}]{\Ent[\nu_{s_0}]{f}}}}}_{\!{(C)}}.
    \end{align*}
    
    Note that when \Cref{assump:smooth} holds, by choosing $T=2L$, $\nu_T$ is $L$-strongly log-concave. By a well-known consequence of the Bakry-\'Emery criterion (see \EG, ~\cite[Section 5.7]{BGL14}), it satisfies modified log-Sobolev inequality with constant $\frac{1}{2L}$ almost surely. Therefore, for the term $\!{(A)}$, we have
    \[
        \!{(A)} = \sup_{f:\bb R^d \to \bb R_{>0}} \frac{\E[\xi_T]{\Ent[\nu_T]{f}}}{\E[\xi_T]{\E[\nu_T]{f^{-1}\|\grad f\|^2}}} \leq \frac{1}{2L}.
    \]

    From Lemma~\ref{lem:beforeconcate},
    \[
        \!{(B)} = \sup_{f:\bb R^d \to \bb R_{>0}} \frac{\E[\xi_{s_0}]{\Ent[\nu_{s_0}]{f}}}{\E[\xi_{s_0}]{\E[\xi'_{T-s_0}]{\Ent[\nu'_{T-s_0}]{f}}}}.
    \]
    Consider the Markov chain induced by \Cref{eq:modifiedSL}. According to Lemma~\ref{lem:RGO}, under \Cref{cond:cov}, for any $v\in \bb R^d$,
    \begin{equation*}
        \norm{\cov{\+T_v \nu'_s}}_{\!{op}} \leq 
        \frac{1+s_0+s +L_{s+s_0}}{\tp{1+s_0+s}(s_0+s)}
    \end{equation*}
    From \Cref{thm:mLSIconcervation}, we have
    \begin{align*}
        \sup_{f:\bb R^d \to \bb R_{>0}} \frac{\Ent[\nu_{s_0}]{f}}{\E[\xi'_{T-s_0}]{\Ent[\nu'_{T-s_0}]{f}}} 
        &\leq \exp\set{\int_0^{T-s_0} \frac{1+s_0+s +L_{s+s_0}}{\tp{1+s_0+s}(s_0+s)} \dd s}\\
        &= \exp\set{ \int_0^{T-s_0} \tp{ \frac{1}{s_0+s} + \frac{L_{s+s_0}}{(s_0+s)(1+s_0+s)}} \d s} \\
        &\leq \frac{T}{s_0}\cdot \exp\set{ \int_{s_0}^{T}  \frac{L_{s}}{s(1+s)} \d s},
    \end{align*}
    and consequently,
    \[
        \!{(B)} \leq \frac{2L}{s_0}\cdot \exp\set{ \int_{s_0}^{2L}  \frac{L_{s}}{s(1+s)} \d s}.
    \]
    % Furthermore, the term $\!{(C)}$ is exactly $C^{\!{mLSI}}_{\mu}\tp{\*P^{(s_0)}}$. 
    Combining all these together, the proof is finished.
\end{proof}

\begin{remark}
    Note that in \Cref{thm:concatenation,thm:concate-mLSI}, we establish integral-form upper bounds for the \Poincare and mLSI constants of $\mu$ under \Cref{cond:cov} (derived from \Cref{assump:smoothplus}). Analogous bounds can be derived using Lipschitz transport maps. Specifically, Lemma 3 of~\cite{MS23} implies similar integral bounds via the Kim-Milman flow map (\cite{KM12}) under~\Cref{assump:smoothplus}. However, a limitation of their result is that the upper limit of integration is $\infty$. This corresponds to setting $s_0=0$ in our framework. Crucially, under only \Cref{assump:smoothplus}, this integral diverges, rendering the bound vacuous. 
    
    A key advantage in our work is the explicit truncation at a small $s_0 > 0$. On the one hand, this allows us to bound the PI and mLSI constants under specific settings, such as the stronger moment conditions discussed in subsequent sections. We believe similar truncation can be applied in the transport map approach to establishing functional inequalities. However,  
    on the other hand, our approach directly implies efficient sampling via the late initialization technique in \Cref{sec:algo}, even when the integral diverges and standard functional inequalities fail.
    % \ctodo{Shall we make similar remarks in the intro?}
\end{remark}

\subsection{Bounding the \Poincare constant under stronger moment conditions}\label{sec:sub-G-PI}
Then we show that $C^{\!{PI}}_{\mu}\tp{\*P^{(s)}}$ can be bounded under the stronger moment bound \Cref{assump:momentplus}.

\begin{lemma}\label{lem:initialPI}
    Under \Cref{assump:momentplus}, we have $C^{\!{PI}}_{\mu}\tp{\*P^{(s)}}\leq \frac{1}{2-e^{4s\lambda^2}}$ for any $0\leq s\leq \frac{\log 2}{4\lambda^2}$. 
\end{lemma}
\begin{proof}
   By the law of total variance, for any $f: \bb R^d\to \bb R$,
   \[
        \Var[\mu]{f} = \E[ \xi_{s }]{\Var[\nu_{s}]{f}} + \Var[\xi_{s }]{\E[\nu_{s }]{f}}.
    \]
    Therefore, 
    \begin{equation}
        \frac{\Var[\mu]{f}}{\E[\xi_s]{\Var[\nu_s]{f}}} = \frac{1}{1 - \frac{\Var[\xi_{s}]{\E[\nu_{s}]{f}}}{\Var[\mu]{f}}}, \label{eq:initialPI-1}
    \end{equation}
    and it is sufficient to prove an upper bound for $\frac{\Var[\xi_{s}]{\E[\nu_{s}]{f}}}{\Var[\mu]{f}}$. We assume without loss of generality that $\E[\mu]{f}=0$. Note that 
    \begin{align*}
    \Var[\xi_s]{\E[\nu_s]{f}} 
    &= \int \tp{\int f(x) \nu_s(x|y) \d x}^2  \xi_s(y)\d y\\
    &=\int \tp{\int f(x)\cdot \mu(x)\cdot \frac{p_{X(s)|X}(y|x)}{ \xi_s(y)}\d x}^2  \xi_s(y) \d y\\
    \mr{$\E[\mu]{f}=0$}&=\int \tp{\int f(x)\cdot \mu(x)\cdot \frac{p_{X(s)|X}(y|x)}{ \xi_s(y)}\d x - \int f(x)\cdot \mu(x) \d x}^2  \xi_s(y) \d y\\
    &=\int \tp{\int f(x)\cdot \mu(x)\cdot \tp{\frac{p_{X(s)|X}(y|x)}{ \xi_s(y)}-1}\d x }^2  \xi_s(y) \d y\\
    \mr{Cauchy-Schwarz inequality}
    &\le \int \tp{\int f(x)^2 \mu(x) \d x}\tp{\int \tp{\frac{p_{X(s)|X}(y|x)}{ \xi_s(y)} - 1}^2 \mu(x) \d x}  \xi_s(y) \d y
    \\
    &=\Var[\mu]{f} \cdot \int \tp{\int \tp{\frac{p_{X(s)|X}(y|x)}{ \xi_s(y)} - 1}^2 \mu(x) \d x}  \xi_s(y) \d y.
    \end{align*}

    Therefore, $\frac{\Var[\xi_s]{\E[\nu_s]{f}}}{\Var[\mu]{f}}$ can be upper bounded by 
    \begin{align}
    \int \tp{\int \tp{\frac{p_{X(s)|X}(y|x)}{\xi_s(y)} - 1}^2 \mu(x) \d x} \xi_s(y) \d y &= \int \tp{\int \tp{\frac{p_{X(s)|X}(y|x)}{\xi_s(y)} - 1}^2 \xi_s(y) \d y} \mu(x) \d x.
    %  \notag \\
    % &= \int \chi^2\tp{p_{X(s)|X}(\cdot | x) \| \xi_s} \cdot  \mu(x) \d x \notag \\
    % &= \int \chi^2\tp{\+N\tp{sx, s\cdot I_d} \| \xi_s} \cdot  \mu(x) \d x. 
    \label{eq:6}
    \end{align}
    We remark that the above quantity can be written as $\chi^2\tp{p_{X,X(s)}\,\|\,p_X\otimes p_{X(s)}}$, which is known as the $\chi^2$-mutual information between $X$ and $X(s)$.
    Note that 
    \[
    \xi_s(y) = \int_{z\in \bb R^d} \mu(z)\cdot \+N\tp{y; sz, s\cdot I_d} \dd z = \E[Z\sim \mu]{ \+N\tp{y; sZ, s\cdot I_d}}.
    \]
    Then for any fixed $x\in \bb R^d$,
    \begin{align*}
    \int \tp{\frac{p_{X(s)|X}(y|x)}{\xi_s(y)} - 1}^2 \xi_s(y) \d y &= \int  \frac{\+N\tp{y;sx, s\cdot I_d}^2}{\xi_s(y)} \dd y -1\\
    &= \int  \frac{\+N\tp{y;sx, s\cdot I_d}^2}{\E[Z\sim \mu]{ \+N\tp{y; sZ, s\cdot I_d}}} \dd y -1\\
    \mr{Jensen inequality} &\leq \mathlarger{\int}  \+N\tp{y;sx, s\cdot I_d}^2\cdot \E[Z\sim \mu]{\frac{1}{ \+N\tp{y; sZ, s\cdot I_d}}} \dd y -1 \\
    &= \E[Z\sim \mu]{\int \frac{\+N\tp{y;sx, s\cdot I_d}^2}{ \+N\tp{y; sZ, s\cdot I_d}} \dd y - 1}\\
    &= \E[Z\sim \mu]{ \chi^2\tp{\+N\tp{sx, s\cdot I_d}\,\|\,\+N\tp{sZ, s\cdot I_d} }}.
    \end{align*}
    From Lemma~\ref{lem:chi2ofGaussian}, $\chi^2\tp{\+N\tp{sx, s\cdot I_d}\,\|\,\+N\tp{sz, s\cdot I_d} } = e^{s\norm{x-z}^2}-1$.
    Bringing these into \cref{eq:6}, we have
    \begin{align*}
    \frac{\Var[Y_s\sim\xi_s]{\E[\nu_s]{f}}}{\Var[\mu]{f}} &\leq \E[X,Z\sim \mu]{e^{s\|X-Z\|^2}} -1\\
    &\leq \E[X,Z\sim \mu]{e^{2s\|X\|^2 + 2s\|Z\|^2}} -1 \\
    &= \E[X\sim \mu]{e^{2s\|X\|^2}}^2 -1\\
    \mr{\Cref{assump:momentplus}}&\leq e^{4s\lambda^2} -1.
    \end{align*}
    When $0\leq s<\frac{\log 2}{4\lambda^2}$, the above value is smaller than $1$ and the final result can be derived via plugging this bound into \cref{eq:initialPI-1}.
\end{proof}

Then the \Poincare inequality is a direct corollary of \Cref{thm:concatenation} and Lemma~\ref{lem:initialPI}.
\begin{theorem}\label{thm:PI}
    If a distribution satisfies \Cref{assump:smooth}, \ref{assump:smoothplus} and \ref{assump:momentplus}, then
    \[
        C^{\!{PI}}_{\mu}\leq \min_{s_0\in \left(0,\frac{\log 2}{4\lambda^2}\right]} \frac{1}{2-e^{4s_0\lambda^2}}\cdot \frac{2}{s_0}\cdot \exp\set{ \int_{s_0}^{2L}  \frac{L_{s}}{s(1+s)} \d s},
    \] 
    where $L_s = L_t^{\!{OU}}$ with $t = \log\sqrt{\frac{s+1}{s}}$.
\end{theorem}
% \ctodo{Change corollary to theorem}

\subsection{Bounding the modified log-Sobolev constant under stronger moment conditions}\label{sec:sub-G-mLSI}

We first prove a defective modified log-Sobolev inequality in \Cref{sec:defective-MLSI}. Then we sketch how to boost it to a standard modified log-Sobolev inequality in \Cref{sec:boosting-MLSI}. Our approach is inspired by the proof of Theorem 1 in \cite{CCNW21}.
% \htodo{mention \cite{CCNW21}}

\subsubsection{A Defective modified log-Sobolev inequality}\label{sec:defective-MLSI}

In order to prove a defective modified log-Sobolev inequality for $\*P^{(s_0)}$, we will use the following lemmas.

\begin{lemma}[Lemma 1 in \cite{CCNW21}]\label{lem:ent-chi2}
    Let $\pi$ and $\rho$ be probability distributions. For any positive function $f$,
    \[
      \E[\pi]f \log \frac{\E[\pi]f}{\E[\rho]f}\le \Ent[\pi]f+\E[\pi]f \log(1+\chi^2(\pi\|\rho)).
    \]
\end{lemma}

\begin{lemma}[Variational formula for entropy, a corollary of Lemma 3.15 in \cite{Van16}]\label{lem:variat-ent}
    For any random variable $Z$ satisfying $\E{\exp(Z)}\le 1$, $\E{YZ}\le \Ent{Y}$.
\end{lemma}

Then we are ready to prove a defective modified log-Sobolev inequality for $\*P^{(s_0)}$.

\begin{proposition}\label{prop:defective-MLSI-s0}
    Under \Cref{assump:momentplus}, for any $0\le s\le \frac 1{12\lambda^2}$ and $f:\bb R^d \to \bb R_{>0}$,
    \[
      \Ent[\mu]f\le 3\cdot\E[\xi_s]{\Ent[\nu_s]f}+24s\lambda^2\cdot \E[\mu]f.
    \]
\end{proposition}
\begin{proof}
    By Lemma~\ref{lem:ent-chi2},
    \begin{align}
      \Ent[\xi_s]{\E[\nu_s]f}=\E[\xi_s]{\E[\nu_s]f \log\frac{\E[\nu_s]f}{\E[\mu]f}}\le \E[\xi_s]{\Ent[\nu_s]f+\E[\nu_s]f\log(1+\chi^2(\nu_s\|\mu))}.\label{eq:mls1}
    \end{align}
    Now we upper bound $\E[\xi_s]{(1+\chi^2(\nu_s\|\mu))^2}$ .
    \begin{align*}
        \E[\xi_s]{(1+\chi^2(\nu_s\|\mu))^2} &= \int \tp{\int\tp{\frac{\nu_s(x|y)}{\mu(x)}}^2\mu(x)\d x}^2\xi_s(y)\d y\\
        \mr{Jensen inequality} &\le \int \tp{\int \tp{\frac{\nu_s(x|y)}{\mu(x)}}^4\mu(x)\d x}\xi_s(y)\d y\\
        &= \int \tp{\int \tp{\frac{p_{X(s)|X}(y|x)}{\xi_s(y)}}^4\mu(x)\d x}\xi_s(y)\d y\\
        &= \int \tp{\int \tp{\frac{p_{X(s)|X}(y|x)}{\xi_s(y)}}^4\xi_s(y)\d y}\mu(x)\d x.
    \end{align*}
    Note that $p_{X(s)|X}(y|x)=\+N(y;sx,s\cdot \Id)$ and $\xi_s(y)=\E[Z\sim\mu]{\+N(y;sZ,s\cdot\Id)}$. Then for any fixed $x\in \bb R^d$,
    \begin{align*}
        \int \tp{\frac{p_{X(s)|X}(y|x)}{\xi_s(y)}}^4\xi_s(y)\d y &= \int \frac{\+N(y;sx,s\cdot\Id)^4}{\xi_s(y)^3}\d y\\
        &= \int \frac{\+N(y;sx,s\cdot\Id)^4}{\E[Z\sim\mu]{\+N(y;sZ,s\cdot\Id)}^3}\d y\\
        \mr{Jensen inequality} &\le \int \+N(y;sx,s\cdot\Id)^4\cdot{\E[Z\sim\mu]{\+N(y;sZ,s\cdot\Id)^{-3}}}\d y\\
        &= \E[Z\sim\mu]{\int \frac{\+N(y;sx,s\cdot\Id)^4}{\+N(y;sZ,s\cdot\Id)^3}\d y}.
    \end{align*}
    We have
    \begin{align*}
        \E[\xi_s]{(1+\chi^2(\nu_s\|\mu))^2}&\le \int \E[Z\sim\mu]{\int \frac{\+N(y;sx,s\cdot\Id)^4}{\+N(y;sZ,s\cdot\Id)^3}\d y}\mu(x)\d x\\
        &= \E[X,Z\sim\mu]{\int \frac{\+N(y;sX,s\cdot\Id)^4}{\+N(y;sZ,s\cdot\Id)^3}\d y}\\
        \mr{Lemma~\ref{lem:chi4ofGaussian}} &= \E[X,Z\sim\mu]{e^{6s\|X-Z\|^2}}\\
        &\le \E[X,Z\sim\mu]{e^{12s\|X\|^2+12s\|Z\|^2}}\\
        &= \E[X\sim\mu]{e^{12s\|X\|^2}}^2\\
        \mr{\Cref{assump:momentplus}} &\le e^{24s\lambda^2}.
    \end{align*}
    Therefore $\E[\xi_s]{\exp(2\log(1+\chi^2(\nu_s\|\mu))-24s\lambda^2)}\le 1$, by Lemma~\ref{lem:variat-ent},
    \[
        \E[\xi_s]{\tfrac 12\cdot\E[\nu_s]f \cdot (2\log(1+\chi^2(\nu_s\|\mu))-24s\lambda^2)}\le \Ent[\xi_s]{\tfrac 12\cdot\E[\nu_s]f}=\tfrac 12\cdot\Ent[\xi_s]{\E[\nu_s]f}.
    \]
    Bringing this into \cref{eq:mls1},
    \begin{align}
        \Ent[\xi_s]{\E[\nu_s]f}\le \E[\xi_s]{\Ent[\nu_s]f}+\tfrac 12\cdot\Ent[\xi_s]{\E[\nu_s]f}+12s\lambda^2\cdot \E[\mu]f.\label{eq:mls2}
    \end{align}
    By the law of total entropy, for any $f: \bb R^d\to \bb R_{>0}$,
    \[
        \Ent[\mu]{f}=\E[\xi_s]{\Ent[\nu_s]{f}}+\Ent[\xi_s]{\E[\nu_s]{f}}.
    \]
    The final result can be derived via plugging this into \cref{eq:mls2}.
\end{proof}

\subsubsection{From defective to standard modified log-Sobolev inequality}\label{sec:boosting-MLSI}

Combining the above defective modified log-Sobolev inequality for $\*P^{(s_0)}$ and similar arguments in the proof of \Cref{thm:concate-mLSI}, we obtain defective classical modified log-Sobolev inequality for $\mu$.

\begin{corollary}\label{cor:defective-MLSI-mu}
    Under \Cref{assump:momentplus}, for any $0\le s_0\le \frac 1{12\lambda^2}$ and $f:\bb R^d \to \bb R_{>0}$,
    \[
      \Ent[\mu]f\le \frac{3}{s_0}\cdot \exp\set{\int_{s_0}^{2L}\frac{L_s}{s(1+s)}\d s}\cdot\E[\mu]{f^{-1}\|\grad f\|^2}+24s_0\lambda^2\cdot \E[\mu]f.
    \]
\end{corollary}

On the other hand, we know that $\mu$ satisfies \Poincare inequality. By a standard tightening argument, a defective mLSI together with \Poincare inequality implies mLSI.

\begin{proposition}[Proposition 5.1.3 of \cite{BGL14}]\label{prop:tighten}
    If a distribution $\mu$ over $\bb R^d$ satisfies classical \Poincare inequality with constant $C^{\!{PI}}_\mu$ and a defective mLSI $\Ent[\mu]f \le A\cdot \E[\mu]{f^{-1}\|\grad f\|^2}+B\cdot \E[\mu]f$ for any $f:\bb R^d\to \bb R_{>0}$, then
    \[
        C^{\!{mLSI}}_{\mu}\le A+\frac {B+2}4 C^{\!{PI}}_\mu.
    \]
\end{proposition}

Then the modified log-Sobolev inequality is a direct corollary of \Cref{thm:PI}, Corollary~\ref{cor:defective-MLSI-mu} and Proposition~\ref{prop:tighten}.

\begin{theorem}\label{thm:mLSI}
    If a distribution satisfies \Cref{assump:smooth}, \ref{assump:smoothplus} and \ref{assump:momentplus}, then
    \[
        C^{\!{mLSI}}_{\mu}\le \min_{s_0\in \left(0,\frac 1{12\lambda^2}\right]}\frac{1}{s_0}\tp{3+\frac{12s_0\lambda^2+1}{2-e^{4s_0\lambda^2}}}\cdot \exp\set{\int_{s_0}^{2L}\frac{L_s}{s(1+s)}\d s}
    \] 
    where $L_s = L_t^{\!{OU}}$ with $t = \log\sqrt{\frac{s+1}{s}}$.
\end{theorem}

\subsection{The modified log-Sobolev constant of the mixture of strongly log-concave and log-smooth distributions}\label{sec:mixofSLC}

In this section, we apply our method to study the modified log-Sobolev constant of a mixture of strongly log-concave and log-smooth distributions. 
% especially the mixture of Gaussian distributions, which is perhaps the gold standard in the study of non-log-concave distributions. We will demonstrate that our method for establishing \Cref{thm:concatenation}\footnote{The argument in this section is essentially the same as that for \Cref{thm:concatenation} with choice of parameters tailored for mixture.} can be used to derive [sharper] bounds for modified log-Sobolev constants of mixtures of Gaussian distributions than previous works.
% \htodo{Mention arxiv 2201.01382.}

To be precise, let $\rho$ denote a probability distribution over $\bb R^d$ with density $\propto e^{-U}$. Assume its potential function $U:\bb R^d\to \bb R$ is twice differentiable, $m$-strongly convex and $L$-smooth. The convolution $\mu = \nu * \rho$, with density $\mu(y) = \int_{\bb R^d} \rho(y-x)\nu(\d x)$, is a mixture of translations of $m$-strongly log-concave and $L$-log-smooth distribution $\rho$. When $\rho$ is the Gaussian distribution $\+N(0,\Sigma)$, $\mu$ is a mixture of Gaussian distributions with the same covariance. In the following part, we define $\+B_R = \set{x\in \bb R^d\cmid \|x\|\leq R}$ as the $d$-dimensional ball with radius $R$, and for a vector $v\in \bb R^d$, define $v^{\otimes 2} = vv^{\top}$. We will consider the case where $\nu$ is supported inside $\+B_R$ for some $R>0$.

Consider two distributions $\xi$ and $\pi$ where $\xi(y) \propto \int_{\bb R^d} \pi(x)\cdot e^{-f(x,y)}\dd x$ for some function $f:\bb R^{2d}\to \bb R$. Let $\nu_y$ be the distribution with density $\nu_y(x)\propto \pi(x)\cdot e^{-f(x,y)}$.  The next lemma is a generalization of Lemma~\ref{lem:hessian-pt} and its proof is provided in \Cref{sec:concat-pf}.
\begin{lemma}\label{lem:hessian-2}
    $-\grad^2_y \log \xi(y) = \E[X\sim \nu_y]{\grad^2_y f(X,y)} - \Cov[X\sim \nu_y]{\grad_y f(X,y)}$ for any $y\in\bb R^d$.
\end{lemma}

Equipped with this lemma, we can bound the log-smoothness of the mixture distribution.

\begin{lemma}\label{lem:mix-smooth}
    Suppose $\rho$ is an $m$-strongly log-concave and $L$-log-smooth distribution and $\nu$ is supported inside $\+B_R$. The distribution $\mu =  \nu * \rho$ satisfies $-\grad^2_y \log \mu(y)\mge (m-L^2R^2)\cdot\Id$ for any $y\in\bb R^d$.
\end{lemma}
% \htodo{Maybe this proof should be provided in appendix?}
\begin{proof}
    For any fixed $y\in\bb R^d$ , choosing $f(x,y)=U(y-x)$ and $\pi=\nu$ in Lemma~\ref{lem:hessian-2}, we have
    \[
        -\grad^2_y \log \mu(y) = \E[X\sim \nu_y]{\grad^2 U(y-X)} - \Cov[X\sim \nu_y]{\grad U(y-X)}.
    \]
    Since $U$ is $m$-strongly convex, $\grad^2 U(x)\mge m\cdot\Id$ for all $x\in\bb R^d$. For the covariance term,
    \begin{align*}
        \norm{\Cov[X\sim \nu_y]{\grad U(y-X)}}_{\!{op}} &= \sup_{v\in\bb R^d,\|v\|=1} v^{\top}\Cov[X\sim \nu_y]{\grad U(y-X)}v \\
        &= \sup_{v\in\bb R^d,\|v\|=1} \E[X\sim \nu_y]{v^{\top}\tp{\grad U(y-X)-\E[X'\sim\nu_y]{\grad U(y-X')}}^{\otimes 2}v}\\
        &= \sup_{v\in\bb R^d,\|v\|=1} \Var[X\sim \nu_y]{v^{\top}\grad U(y-X)}\\
        \mr{define $g(x)=v^{\top}\grad U(y-x)$}&= \sup_{v\in\bb R^d,\|v\|=1} \frac 12\cdot\E[X,X'\sim \nu_y]{\tp{g(X)-g(X')}^2}\\
        \mr{$g(x)$ is $L$-Lipschitz}&\le \sup_{v\in\bb R^d,\|v\|=1} \frac {L^2}2\cdot\E[X,X'\sim \nu_y]{\|X-X'\|^2}\\
        &=L^2\cdot \tp{\E[X\sim \nu_y]{\|X\|^2}-\norm{\E[X\sim \nu_y]{X}}^2}\\
        \mr{$\nu_y$ is supported on $\+B_R$}&\le L^2R^2.
    \end{align*}
    Therefore,
    \[
        -\grad^2_y \log \mu(y)\mge m\cdot\Id-L^2R^2\cdot\Id=(m-L^2R^2)\cdot\Id.
    \]
\end{proof}

\subsubsection{The convolution of OU process}
Consider the OU process $\set{X^{\OU}(t)}_{t\geq 0}$ starting from $\mu=\nu * \rho$. Recall that we use $\xi^{\OU}_t$ to denote the law of $X^{\OU}(t)$.  We now analyze the convolution of $\xi^{\OU}_t$. Let $W\sim\nu$, $Z\sim\rho$ and $\zeta\sim\+N(0,\Id)$ be mutually independent. Then $X^{\OU}(t)$ equals $e^{-t}(W+Z)+\sqrt{1-e^{-2t}}\cdot\zeta$ in distribution. Note that this can be viewed as $e^{-t}W+(e^{-t}Z+\sqrt{1-e^{-2t}}\cdot\zeta)$. Let $\nu_t$ be the law of $e^{-t}W$ and $\rho_t$ be the law of $(e^{-t}Z+\sqrt{1-e^{-2t}}\cdot\zeta)$. Then $\xi^{\OU}_t=\nu_t * \rho_t$. By the assumption on $\nu$, $\nu_t$ is supported on $\+B_{e^{-t}R}$. The following lemma establishes bounds on the strong log-concavity and log-smoothness of $\rho_t$.

\begin{lemma}[Lemma 28 in \cite{LPSR21}]\label{lem:m-gaussian1}
      Suppose $\pi$ is a distribution on $\bb R^d$ such that $\forall x\in\bb R^d,M_{1}^{-1} \preceq -\grad^2 \log \pi(x) \preceq M_2^{-1}$ for some $M_1,M_2\in \bb R^{d\times d},M_1,M_2\succ 0$. Let $\xi$ be the mixture distribution $\pi*\+N(0,M)$. Then for any $x\in \bb R^d$,
      \[
        (M_1+M)^{-1} \preceq  -\grad^2 \log \xi(x) \preceq (M_2+M)^{-1}.
      \]
\end{lemma}

\begin{corollary}\label{cor:OUsmooth}
    The distribution $\rho_t$ satisfies
    \[
      \frac{m}{e^{-2t}+(1-e^{-2t})m}\cdot\Id \mle -\grad^2 \log \rho_t(x)\mle \frac{L}{e^{-2t}+(1-e^{-2t})L}\cdot\Id
    \]
    for any $t\ge 0$ and $x\in\bb R^d$.
\end{corollary}
\begin{proof}
    Recall that $Z\sim\rho$ where $\rho$ is $m$-strongly log-concave and $L$-log-smooth. Therefore for any fixed $t\ge 0$, $e^{2t}m \mle-\nabla^2 \log p_{e^{-t}Z}(x)\mle e^{2t}L$ for any $x\in \bb R^d$. From Lemma~\ref{lem:m-gaussian1}, with $M_1=e^{-2t}m^{-1}\cdot\Id$, $M_2=e^{-2t}L^{-1}\cdot\Id$ and $M=(1-e^{-2t})\Id$ , we have
    \[
      \frac{m}{e^{-2t}+(1-e^{-2t})m}\cdot\Id \mle -\grad^2 \log \rho_t(x)\mle \frac{L}{e^{-2t}+(1-e^{-2t})L}\cdot\Id
    \]
    for any $x\in \bb R^d$.
\end{proof}

\subsubsection{The bound on the mLSI constant via concatenation}
Now we are ready to prove the modified log-Sobolev constant of the mixture of strongly log-concave and log-smooth distributions.
\begin{theorem}\label{thm:mLSIofmix}
    Suppose $\rho$ is an $m$-strongly log-concave and $L$-log-smooth distribution and $\nu$ is supported inside $\+B_R$. The modified log-Sobolev constant of mixture distribution $\mu=\nu*\rho$ satisfies
    \[
      C^{\!{mLSI}}_{\mu}\leq \frac 1{2m}\cdot e^{LR^2}.
    \]
\end{theorem}
\begin{proof}
    Fix any $t\ge 0$. Recall that $\xi^{\OU}_t=\nu_t * \rho_t$ and $\nu_t$ is supported on $\+B_{R'}$ where $R'\defeq e^{-t}R$. By Corollary~\ref{cor:OUsmooth}, define $m'\defeq \frac{m}{e^{-2t}+(1-e^{-2t})m}$ and $L'\defeq \frac{L}{e^{-2t}+(1-e^{-2t})L}$, then $\rho_t$ is $m'$-strongly log-concave and $L'$-log-smooth. From Lemma~\ref{lem:mix-smooth}, for any $y\in\bb R^d$,
    \[
      -\grad^2_y \log \xi_t^{\OU}(y)\mge (m'-L'^2R'^2)\cdot\Id=\tp{\frac{m}{e^{-2t}+(1-e^{-2t})m}-\tp{\frac{L}{e^{-2t}+(1-e^{-2t})L}}^2\tp{e^{-t}R}^2}\cdot\Id.
    \]
    % \htodo{only the semi-log-concave side}
    Therefore, $\mu$ satisfies \Cref{assump:smoothplus} with $L^{\OU}_t=-\frac{m}{e^{-2t}+(1-e^{-2t})m}+\tp{\frac{L}{e^{-2t}+(1-e^{-2t})L}}^2\tp{e^{-t}R}^2$. By Lemma~\ref{lem:cond1}, for any $s\ge 0$ and $v\in\bb R^d$,
    \[
      \norm{\cov{\+T_v \nu_s}}_{\!{op}}\le \frac 1s+\frac{L_s}{s(1+s)}=\frac 1{s+m}+\frac{L^2R^2}{(s+L)^2},
    \]
    where $L_s=L^{\OU}_t$ with $t=\log\sqrt{\frac{s+1}{s}}$. From \Cref{thm:mLSIconcervation}, for any $T>0$,
    \[
      \sup_{f:\bb R^d \to \bb R_{>0}} \frac{\Ent[\mu]{f}}{\E[\xi_T]{\Ent[\nu_T]{f}}} \leq \exp\set{\int_0^T \tp{\frac 1{s+m}+\frac{L^2R^2}{(s+L)^2}} \d s}=\frac{T+m}m\cdot\exp\tp{\frac{LT}{L+T}R^2}.
    \]
    For the classical modified log-Sobolev inequality, we have
    \begin{align*}
      C^{\!{mLSI}}_{\mu} &= \sup_{f:\bb R^d \to \bb R_{>0}} \frac{\Ent[\mu]{f}}{\E[\mu]{f^{-1}\|\grad f\|^2}}  \\
        &= \sup_{f:\bb R^d \to \bb R_{>0}} \frac{\E[\xi_T]{\Ent[\nu_T]{f}}}{\E[\mu]{f^{-1}\|\grad f\|^2}}\cdot \frac{\Ent[\mu]{f}}{\E[\xi_T]{\Ent[\nu_T]{f}}} \\
        &\leq \tp{\sup_{f:\bb R^d \to \bb R_{>0}} \frac{\E[\xi_T]{\Ent[\nu_T]{f}}}{\E[\xi_T]{\E[\nu_T]{f^{-1}\|\grad f\|^2}}}}\cdot \tp{\sup_{f:\bb R^d \to \bb R_{>0}} \frac{\Ent[\mu]{f}}{\E[\xi_T]{\Ent[\nu_T]{f}}}}.
    \end{align*}
    For any $T>\max\set{0,L^2R^2-m}$, $\nu_T$ is $(T+m-L^2R^2)$-strongly log-concave almost surely. By a well-known consequence of the Bakry-\'Emery criterion (see \EG, ~\cite[Section 5.7]{BGL14}), it satisfies modified log-Sobolev inequality with constant $\tfrac 12\cdot(T+m-L^2R^2)^{-1}$ almost surely. Therefore,
    \[
      C^{\!{mLSI}}_{\mu}\leq \inf_{T>\max\set{0,L^2R^2-m}} \frac{T+m}{2m(T+m-L^2R^2)}\cdot\exp\tp{\frac{LT}{L+T}R^2}\le \frac 1{2m}\cdot e^{LR^2}.
    \]
    %% \htodo{need citation of Bakry-Emery for LSI}
\end{proof}

As a special case, our \Cref{thm:mLSIofmix} naturally recovers and generalizes the previous results for mixture of Gaussians in \cite{MS23}. Let $\rho$ be the Gaussian distribution $\+N(0,\Sigma)$. Recall that $\lambda_{\min}(\Sigma)$ denotes the minimum eigenvalue of $\Sigma$. The Gaussian $\+N(0,\Sigma)$ is $\norm{\Sigma}_{\!{op}}^{-1}$-strongly log-concave and $\lambda_{\min}(\Sigma)^{-1}$-log-smooth.
\begin{corollary}
    Suppose $\nu$ is supported on $\+B_R$ and $\mu = \nu * \+N(0, \Sigma)$. We have $C^{\!{mLSI}}_\mu\leq \frac 12\norm{\Sigma}_{\!{op}}\cdot e^{\lambda_{\min}(\Sigma)^{-1}\cdot R^2}$.
\end{corollary}

\bibliographystyle{alpha}
\bibliography{arxiv}

\appendix

\section{The convergence of the continuous-time Gaussian dynamics}\label{sec:GD-pf}

We first prove the convergence of the continuous-time Gaussian dynamics with the heat kernel $\*H_t$.
The following lemma is proved in \cite{LP17} for chains with a countable space. We generalize their proof to space $\bb R^d$.
%\htodo{Define some }
\begin{proposition}[a variant of Lemma 20.5 in \cite{LP17}]\label{prop:Ht}
    Assume $\*P^{(T)}$ is a reversible and irreducible Markov chain, and the stationary distribution $\mu$ satisfies a \Poincare inequality with regard to $\*P^{(T)}$ with constant $C$. Then for any bounded measurable $f:\bb R^d\to \bb R$ with $\+E_{\*P^{(T)}}(f)\neq 0$ and $\E[\mu]{f}=0$,
    \[
        \|\*H_t f\|^2_{\mu} \leq e^{-\frac{2t}{C}}   \cdot \Var[\mu]{f}.     
    \]
\end{proposition}
%\ctodo{The family of $f$.}
\begin{proof}
    Recall the definition of $\*H_t$ in \cref{eq:transition}, $ \forall t\geq 0, \forall x,y\in \bb R^d$, the density
    \begin{align*}
         \*H_t(x,y) &= \sum_{k=0}^{\infty} \Pr{K_t=k}\cdot \tp{\*P^{(T)}}^k(x, y)\\
         &= \sum_{k=0}^{\infty} \frac{e^{-t}t^k}{k!}\cdot \tp{\*P^{(T)}}^k(x, y).
    \end{align*}
    Then 
%    \htodo{It is exchangable, right?}
    \begin{align*}
        \frac{\dd }{\dd t} \*H_t(x,y) &= \sum_{k=1}^{\infty} \frac{e^{-t}t^{k-1}}{(k-1)!}\cdot \tp{\*P^{(T)}}^k(x, y) - \sum_{k=0}^{\infty} \frac{e^{-t}t^k}{k!}\cdot \tp{\*P^{(T)}}^k(x, y) \\
        &= \int_{\bb R^d} \*P^{(T)}(x,z)\cdot \*H_t(z,y) \dd z - \*H_t(x,y),
    \end{align*}
    and
    \begin{align}
        \frac{\dd }{\dd t} \*H_t f(x) &= \int_{\bb R^d} \*P^{(T)}(x,z)\cdot \*H_t f(z) \dd z - \*H_t f(x) \notag \\
        &= \tp{\*P^{(T)} - I}(\*H_t f)(x).\label{eq:Ht-1}
    \end{align} 
    Define the function $g(t) = \|\*H_t f\|^2_{\mu} $. From \Cref{eq:Ht-1},
    \begin{align*}
        g'(t) &= \int_{\bb R^d} 2 \*H_t f(x)\cdot \frac{\dd }{\dd t} \*H_t f(x) \cdot \mu(x) \dd x\\
        &= \int_{\bb R^d} 2 \*H_t f(x)\cdot \tp{\*P^{(T)} - I}(\*H_t f)(x) \cdot \mu(x) \dd x\\
        &= 2\int_{\bb R^d \times \bb R^d} \*H_t f(x) \mu(x)  \cdot \*H_t f(y) \*P^{(T)}(x,y) \dd y \dd x  - 2\int_{\bb R^d} \*H_t f(x)^2 \cdot \mu(x) \dd x\\
        &= -\int_{\bb R^d \times \bb R^d} \tp{\*H_t f(x)-\*H_t f(y)}^2 \mu(x)\*P^{(T)}(x,y) \dd y \dd x \\
        &= -2\+E_{\*P^{(T)}}(\*H_t f).
    \end{align*}
    % It is easy to verify that the stationary distribution of the Markov process with heat kernel $\set{\*H_t}$ is $\mu$. Therefore, $\E[\mu]{\*H_t f} = \E[\mu]{f}=0$.
    Due to the reversibility of $\*H_t$,
    \[
        \E[\mu]{\*H_t f} = \int_{\bb R^d\times \bb R^d} f(y)\*H_t(x,y)\mu(x)\dd x \dd y = \int_{\bb R^d\times \bb R^d} f(y)\*H_t(y,x)\mu(y)\dd x \dd y = \E[\mu]{f}=0.
    \]
    Since $\mu$ satisfies a \Poincare inequality with regard to $\*P^{(T)}$ with constant $C$, 
    \[
        \+E_{\*P^{(T)}}(\*H_t f) \geq \frac{\Var[\mu]{\*H_t f}}{C} = \frac{\|\*H_t f\|^2_{\mu}}{C}.
    \]
    Therefore, $g'(t)\leq -\frac{2}{C} g(t)$. Consequently, 
    \[
        \|\*H_t f\|^2_{\mu}=g(t)\leq e^{-\frac{2t}{C}} \cdot g(0) = e^{-\frac{2t}{C}}\cdot \Var[\mu]{f}.
    \]
\end{proof}

Let $\mu_t$ be the distribution induced by running the continuous-time restricted Gaussian dynamics with heat kernel $\*H_t$ starting from $\mu_0$. Lemma~\ref{lem:Htconvergence} gives the convergence rate of this Markov process in $\chi^2$ divergence and total variation distance.
\begin{lemma}\label{lem:Htconvergence}
    Assume $\*P^{(T)}$ is a reversible and irreducible Markov chain, and the stationary distribution $\mu$ satisfies a \Poincare inequality with regard to $\*P^{(T)}$ with constant $C$. Then $\chi^2(\mu_t\,\|\, \mu) \leq e^{-\frac{2 t}{C}} \cdot \chi^2(\mu_0\,\|\, \mu)$, and consequently, $\!{TV}(\mu_t,\mu)\leq \tfrac{1}{2}\cdot e^{-\frac{t}{C}}\cdot \sqrt{\chi^2(\mu_0\,\|\, \mu)}$.
\end{lemma}
\begin{proof}
    Define the bounded and measurable function $f = \frac{\mu_0}{\mu} -1$. Then 
    \[
        \Var[\mu]{f} = \mathlarger{\int}_{\bb R^d} \tp{\frac{\mu_0(x)}{\mu(x)} -1}^2 \mu(x) \dd x = \chi^2(\mu_0\,\|\, \mu),
    \]
    and
    \begin{align*}
        \*H_t f(x) &= \mathlarger{\int}_{\bb R^d} \tp{\frac{\mu_0(y)}{\mu(y)} -1} \*H_t(x,\dd y) \\
        &= \mathlarger{\int}_{\bb R^d} \frac{\mu_0(y)}{\mu(x)}\cdot \frac{\mu(x)\*H_t(x,y)}{\mu(y)} \dd y  - 1 \\
        \mr{reversibility} &= \mathlarger{\int}_{\bb R^d} \frac{\mu_0(y)}{\mu(x)}\cdot \frac{\mu(y)\*H_t(y, x)}{\mu(y) }\dd y  - 1 \\
        &= \frac{\mu_t(x)}{\mu(x)} - 1.
    \end{align*}
    From Proposition~\ref{prop:Ht}, we have $\chi^2(\mu_t\,\|\, \mu) \leq e^{-\frac{2t}{C}} \cdot \chi^2(\mu_0\,\|\, \mu)$.

    From the Cauchy-Schwarz inequality,
    \[
        \!{TV}(\mu_t,\mu) = \frac{1}{2}\mathlarger{\int}_{\bb R^d} \abs{\mu_t(x)-\mu(x)} \frac{\sqrt{\mu(x)}}{\sqrt{\mu(x)}} \dd x \leq \frac{1}{2} \sqrt{\mathlarger{\int}_{\bb R^d} \frac{\tp{\mu_t(x)-\mu(x)}^2}{\mu(x)} \dd x} = \frac{\sqrt{\chi^2(\mu_t\,\|\, \mu)}}{2}.
    \]
    This yields the second conclusion.
\end{proof}

Then we are ready to prove \Cref{thm:GD}.
\begin{proof}[Proof of \Cref{thm:GD}]
    In the continuous-time Markov process with heat kernel $\set{\*H_t}$, let $\tau\geq 0$ be the time that the $K$-th transition happens. Then $p_{Y_{\+K}}$ is exactly $\mu_{t\wedge \tau}$ with $t=\frac{K}{2}$. According to Lemma~\ref{lem:poistail},
    \begin{equation}
        \Pr{\tau\leq t} = \Pr{K'\geq K} \leq e^{-\frac{t^2}{2K}} \leq e^{-\frac{K}{8}}.  \label{eq:GD-1}
    \end{equation}
    Let $\mu_t^{(>)}$ and $\mu_t^{(\leq)}$ be the distribution $\mu_t$ conditioned on $\tau> t$ and $\tau\leq t$ respectively. Then for any $x\in \bb R^d$, $\mu_t(x) = \Pr{\tau> t}\cdot \mu_t^{(>)}(x) + \Pr{\tau\leq t}\cdot \mu_t^{(\leq)}(x)$.
    Note that
    \begin{align*}
        \!{TV}(\mu_t,\mu) &= \frac{1}{2}\int_{\bb R^d} \abs{\mu_t(x)-\mu(x)} \dd x\\
        \mr{triangle inequality}&\geq \Pr{\tau> t}\cdot \!{TV}\tp{\mu_t^{(>)}, \mu} - \Pr{\tau\leq t}\cdot \!{TV}\tp{\mu_t^{(\leq)}, \mu} \\
        \mr{\Cref{eq:GD-1}} &\geq \Pr{\tau> t}\cdot \!{TV}\tp{\mu_t^{(>)}, \mu}  -  e^{-\frac{K}{8}}.
    \end{align*}
    From Lemma~\ref{lem:Htconvergence}, $\!{TV}(\mu_t,\mu)\leq \frac{e^{-\frac{K}{C^{\!{PI}}_{\mu}\tp{\*P^{(T)}}}} }{2}\cdot \sqrt{\chi^2(\mu_0\,\|\, \mu)} $. This yields 
    \begin{equation}
        \Pr{\tau> t}\cdot \!{TV}\tp{\mu_t^{(>)}, \mu} \leq \frac{e^{-\frac{K}{C^{\!{PI}}_{\mu}\tp{\*P^{(T)}}}}}{2} \cdot \sqrt{\chi^2(\mu_0\,\|\, \mu)} + e^{-\frac{K}{8}}. \label{eq:GD-2}
    \end{equation}
    Combining \Cref{eq:GD-1,eq:GD-2}, 
    \begin{align*}
        \!{TV}\tp{p_{Y_{\+K}}, \mu} &= \!{TV}\tp{\mu_{t\wedge \tau}, \mu} \\
        \mr{triangle inequality}&\leq \Pr{\tau >t}\cdot \!{TV}\tp{\mu_t^{(>)}, \mu} + \Pr{\tau \leq t}\\
        \mr{\Cref{eq:GD-1}}&\leq \Pr{\tau >t}\cdot \!{TV}\tp{\mu_t^{(>)}, \mu} + e^{-\frac{K}{8}} \\
        \mr{\Cref{eq:GD-2}}&\leq \frac{e^{-\frac{K}{C^{\!{PI}}_\mu\tp{\*P^{(T)}}}}}{2} \cdot \sqrt{\chi^2(\mu_0\,\|\, \mu)} + 2e^{-\frac{K}{8}}.
    \end{align*}
 
\end{proof}

\section{The initialization bound in \Cref{algo:modifiedGD}}
In this section, we derive an explicit bound of the term $\log \chi^2\tp{\mu_0\,\|\,\nu_{s_0}(\cdot\,|\,X_{s_0})}$.

% We first show that under \Cref{cond:cov}, with high probability, the second moment of $\nu_{s_0}(\cdot\,|\,X_{s_0})$ is bounded.

% \begin{lemma}\label{lem:Gaussiantail}[A corollary of Theorem B.1 in \cite{Spo23}]
%     Let $X_{s_0}\sim \+N\tp{0, s_0(1+s_0)\Id}$. Then 
%     \[
%         \Pr{\norm{X_{s_0}}^2 \geq 5d\cdot s_0(1+s_0)\cdot \log\frac{1}{\eps_{\!{init}}}} \leq \eps_{\!{init}}.
%     \]
% \end{lemma}

We begin by bounding the first moment of $\nu_{s_0}(\cdot\,|\,y)$.
\begin{lemma}\label{lem:firstm}
    For any $y\in \bb R^d$,  $\E[X\sim \nu_{s_0}(\cdot\,|\,y)]{\norm{X}}\leq 2\sqrt{M}\cdot  e^{\frac{\norm{y}^2}{s_0}+1}$.
\end{lemma}
\begin{proof}
    By definition,
    \begin{align*}
        \E[X\sim \nu_{s_0}(\cdot\,|\,y)]{\norm{X}} &= \frac{\mathlarger{\int}_{\bb R^d} \|x\| \cdot \exp\set{-V(x) - \frac{\norm{s_0\cdot x - y}^2}{2s_0}} \dd x}{\mathlarger{\int}_{\bb R^d} \exp\set{-V(x) - \frac{\norm{s_0\cdot x - y}^2}{2s_0}} \dd x}\\
        &\leq \frac{\mathlarger{\int}_{\bb R^d} \|x\| \cdot e^{-V(x)} \dd x}{\mathlarger{\int}_{\bb R^d} \exp\set{-V(x) - \frac{\norm{s_0\cdot x }^2 + \norm{y}^2}{s_0}} \dd x}\\
        &\leq \E[X\sim \mu]{\norm{X}}\cdot \frac{e^{\frac{\norm{y}^2}{s_0}}\cdot \mathlarger{\int}_{\bb R^d} e^{-V(x)} \dd x}{\mathlarger{\int}_{\bb R^d} \exp\set{-V(x) - \frac{\norm{s_0\cdot x }^2 }{s_0}} \dd x}\\
        &= \E[X\sim \mu]{\norm{X}}\cdot e^{\frac{\norm{y}^2}{s_0}} \cdot \frac{1}{\E[X\sim \mu]{e^{-s_0\norm{x}^2}}}.
    \end{align*}
    From \Cref{assump:moment} and the Jensen inequality, 
    \[
        \E[X\sim \mu]{\norm{X}} \leq \sqrt{\E[X\sim \mu]{\norm{X}^2}} \leq \sqrt{M}.
    \]
    Using \Cref{assump:moment} again and together with the Markov's inequality, we have
    \[
        \Pr[X\sim \mu]{\norm{X}^2 \geq \frac{1}{s_0}} \leq s_0\cdot M <\frac{1}{2} 
    \]
    and therefore,
    \[
        \E[X\sim \mu]{e^{-s_0\norm{x}^2}} \geq \Pr[X\sim \mu]{\norm{X}^2 < \frac{1}{s_0}} \cdot e^{-1} > \frac{1}{2e}.
    \]
    Combining all these together, we have
    \[
        \E[X\sim \nu_{s_0}(\cdot\,|\,y)]{\norm{X}} < 2\sqrt{M}\cdot  e^{\frac{\norm{y}^2}{s_0}+1}.
    \]
\end{proof}

% \begin{lemma}\label{lem:V0}
%     If a distribution $\mu$ with density $\propto e^{-V}$ satisfies \Cref{assump:moment} and $\Cref{assump:smooth}$, then $V(0)-\min V\leq 2LM$.
% \end{lemma}
% \begin{proof}
%     We first prove that $x^* \defeq \arg\min_{x\in \bb R^d} V(x)$ satisfies $\|x^*\|\leq \sqrt{2M}$. Assume this does not hold. Then $\mu(x)<\mu(x^*)$ for any $\|x\|\leq \sqrt{2M}$. Let $R=\sqrt{2M}$ and recall that $\+B_{R}$ denotes the ball $\set{x\in \bb R^d: \|x\|\leq R}$. We have
%     \[
%         \E[\mu]{\|X\|^2}\geq \int_{\+B_{R}} \|x\|^2 \mu(x)\dd x 
%     \]
% \end{proof}
% \htodo{Plug \Cref{lem:V0} into \Cref{lem:initialization} and the main theorem.}

Then we can use the next lemma to bound the initial $\chi^2$ divergence. The proof of Lemma~\ref{lem:initialization} follows the approach of Lemma 32 in \cite{CEL+24}.
\begin{lemma}\label{lem:initialization}
    Define the function $U_y(x) = V(x) + \frac{\norm{y - s_0\cdot x}^2}{2s_0}$. For $\mu_0=\+N\tp{-\frac{\grad U_y(0)}{2(L+s_0)},\frac{I_d}{2(L+s_0)}}$,
    \[
        \log \chi^2\tp{\mu_0\,\|\,\nu_{s_0}(\cdot\,|\,y)} \leq \log(2e) + V(0) - \min V + \frac{d}{2}\log \tp{(L+s_0)\cdot  16e^2M} + \frac{(2d+1)\norm{y}^2}{2s_0}.
    \]
\end{lemma}
\begin{proof}
    By the definition of $\chi^2$ divergence, we know
    \[
        \log \chi^2\tp{\mu_0\,\|\,\nu_{s_0}(\cdot\,|\,y)} \leq \+R_2\tp{\mu_0\,\|\,\nu_{s_0}(\cdot\,|\,y)} \leq \+R_{\infty}\tp{\mu_0\,\|\,\nu_{s_0}(\cdot\,|\,y)}.
    \]
    It only needs to bound
    \begin{align*}
        \+R_{\infty}\tp{\mu_0\,\|\,\nu_{s_0}(\cdot\,|\,y)} &= \log\tp{\sup_{x\in \bb R^d} \frac{\mu_0(x)}{\nu_{s_0}(x|y)}}\\
        &= \log\tp{\sup_{x\in \bb R^d} \frac{\exp\set{-(L+s_0)\cdot \norm{x + \frac{\grad U_y(0)}{2(L+s_0)}}^2 + U_y(x)}}{\tp{\frac{\pi}{L+s_0}}^{\frac{d}{2}}} \cdot \int_{\bb R^d} e^{-U_y(z)} \dd z}.
    \end{align*}

    Recall that 
    \[
        \nu_{s_0}(x\,|\,y) \propto \mu(x)\cdot \exp\set{-\frac{\norm{y - s_0\cdot x}^2}{2s_0}}.
    \]
    Under \Cref{assump:smooth}, $\nu_{s_0}(\cdot\,|\,y)$ is $(L+s_0)$-log-smooth. Then for any $x\in \bb R^d$, there exists some $\lambda\in [0,1]$ such that
    \begin{align*}
        U_y(x) - U_y(0) &= \inner{\grad U_y(\lambda x)}{x} \leq \inner{\grad U_y(0)}{x} + \norm{\inner{\grad U_y(\lambda x) - \grad U_y(0)}{x}}\\
        \mr{smoothness}&\leq \inner{\grad U_y(0)}{x} + (L+s_0)\|x\|^2.
    \end{align*}
    Therefore,
    \begin{align}
        -(L+s_0)\cdot \norm{x + \frac{\grad U_y(0)}{2(L+s_0)}}^2 + U_y(x) &\leq U_y(0) + \inner{\grad U_y(0)}{x} + (L+s_0)\|x\|^2  -(L+s_0)\cdot \norm{x + \frac{\grad U_y(0)}{2(L+s_0)}}^2 \notag \\
        &= U_y(0) - \frac{\norm{\grad U_y(0)}^2}{4(L+s_0)}\\
        &= V(0) + \frac{\norm{y}^2}{2s_0} - \frac{\norm{\grad V(0) - y}^2}{4(L+s_0)}. \label{eq:init-1}
    \end{align}
    Then we calculate $\frac{\int_{\bb R^d} e^{-U_y(z)} \dd z}{\tp{\frac{\pi}{L+s_0}}^{\frac{d}{2}}}$. Let $r = \E[X\sim \nu_{s_0}(\cdot\,|\,y)]{\norm{X}}$. From Lemma~\ref{lem:firstm}, $r\leq 2\sqrt{M}\cdot  e^{\frac{\norm{y}^2}{s_0}+1}$.
    For any $\delta>0$, from the Markov's inequality,
    \begin{align}
        \frac{\int_{\bb R^d} e^{-U_y(z) - \delta\cdot \|z\|^2} \dd z}{\int_{\bb R^d} e^{-U_y(z)} \dd z} &= \E[Z\sim \nu_{s_0}(\cdot\,|\,y)]{e^{- \delta\cdot \|Z\|^2}} \notag \\
        &\geq \exp\set{-4\delta \cdot r^2}\cdot \Pr[Z\sim \nu_{s_0}(\cdot\,|\,y)]{\|Z\| \leq 2r} \notag \\
        &\geq \frac{1}{2}\exp\set{-4\delta \cdot r^2}. \label{eq:init-2}
    \end{align}
    We also have
    \begin{align}
        \frac{\int_{\bb R^d} e^{-U_y(z) - \delta\cdot \|z\|^2} \dd z}{\tp{\frac{\pi}{L+s_0}}^{\frac{d}{2}}} &\leq \frac{e^{-\min U_y}\cdot \int_{\bb R^d} e^{ - \delta\cdot \|z\|^2} \dd z}{\tp{\frac{\pi}{L+s_0}}^{\frac{d}{2}}} \notag \\
        &\leq e^{-\min U_y} \cdot \tp{\frac{L+s_0}{\delta}}^{\frac{d}{2}} \notag \\
        &\leq e^{-\min V} \cdot \tp{\frac{L+s_0}{\delta}}^{\frac{d}{2}}. \label{eq:init-3}
    \end{align}
    Combining \Cref{eq:init-1,eq:init-2,eq:init-3} and choosing $\delta = \frac{1}{4r^2}$, we have
    \begin{align*}
        \+R_{\infty}\tp{\mu_0\,\|\,\nu_{s_0}(\cdot\,|\,y)} &\leq V(0) + \frac{\norm{y}^2}{2s_0} - \frac{\norm{\grad V(0) - y}^2}{4(L+s_0)} + \log(2e) - \min V \\
        &\quad + \frac{d}{2}\log \tp{(L+s_0)\cdot  16M\cdot  e^{\frac{2\norm{y}^2}{s_0}+2}}\\
        &\leq \log(2e) + V(0) - \min V + \frac{d}{2}\log \tp{(L+s_0)\cdot  16e^2M} + \frac{(2d+1)\norm{y}^2}{2s_0}.
    \end{align*}
\end{proof}

% Then we have the following corollary.
% \begin{corollary}\label{coro:initialization-2}
%     In \Cref{algo:modifiedGD}, $\E[X_{s_0}\sim \+N\tp{0, s_0(1+s_0)\Id}]{\log \chi^2\tp{\mu_0\|\nu_{s_0}(\cdot\,|\,X_{s_0})}} $ can be bounded by
%     \[
%         \log(2e) + V(0) - \min V + \frac{d}{2}\log \tp{(L+s_0)\cdot  16e^2M} + \frac{d(2d+1)(1+s_0)}{2}.
%     \]
% \end{corollary}

\section{The implementation of RGO}\label{sec:RGO}
In this section, we provide the details of the rejection sampling algorithm in~\cite{LC23} for the completeness of the paper. Let $V_y^\sigma(x):= V(x) + \frac{1}{2\sigma^2}\|x-y\|^2$. Recall that our target is to generate a sample from the distribution $\mu_{y,\sigma^2}$ with density $\propto \exp\set{-V_y^\sigma(x)}$ under \Cref{assump:smooth}.

To implementation the rejection sampling algorithm, we first need to find an approximate minimizer of $V_y^\sigma$. This can be achieved via the accelerated gradient algorithm in \cite{LC23}.
%  Let $x^*_y = \arg\min_{x\in \bb R^d} V_y^\sigma(x)$.

% \IE,
% \[
%     \norm{\grad V(x) - \grad V(y)} \le L\norm{x-y},\quad\forall x,y \in \bb R^d.
% \]
% This implies that 
% \[
%     \norm{V(x) - V(y) - \left \langle \grad V(y), x-y \right \rangle} \le \frac{L}{2} \norm{x-y}^2,\,\forall x,y \in \bb R^d.
% \]
% Then consider the distribution $\mu_{y, \sigma^2}$ with density $\propto$ $\exp \set{-V(x)-\frac{1}{2\sigma^2}\norm{x-y}^2}$.

% Define
% \[V_y^\sigma(x):= V(x) + \frac{1}{2\sigma^2}\|x-y\|^2\]
% To sample from the distribution with density $\propto \exp\set{-V_y^\sigma(x)}$ when $\sigma$ is small enough, we first need to find the approximate stationary point of $V_y^{\sigma}$. 
% \htodo{$Ld\to 2Ld$.}
\begin{lemma}[A corollary of Proposition 3.2 in \cite{LC23}] \label{lem:AproxSta}
    Assume $V$ is $L$-smooth and $\sigma^2\le \frac{1}{2Ld}$. For any fixed $y\in \bb R^d$, Algorithm~\ref{algo:AproxSta} outputs a point $w\in \bb R^d$ such that
    \[
        \norm{\grad V(w) + \frac{1}{\sigma^2} (w-y)} \le \sqrt{L d}
    \]
    within $\wt{\+O}(1)$
    % within $\+O\tp{\log \tp{Ld\norm{x^*_y}^2}}$
    iterations in expectation.
    % and we call $w\in \bb R^d$ be an approximate stationary point of $V_y^{\sigma}$, i.e.,
    % \[
    %     \norm{\grad V(w) + \frac{1}{\sigma^2} (w-y)} \le \sqrt{L d}
    % \]
    % Then the iteration complexity to find such $w$ is $\tilde{\cal O}(1)$ in Algorithm~\ref{algo:AproxSta}.
\end{lemma}

\begin{algorithm}[H]
    	\caption{Accelerated Gradient Method}
            \label{algo:AproxSta}
	\begin{algorithmic}[1]
		\State Let$\;g(x) = V(x) + \frac{1}{2\sigma^2} \norm{x-y}^2$, let the initial point $ y_0=x_0 = 0$, set $T = \frac{1}{\sigma^2} + L , B = \frac{1}{\sigma^2} - L>0 $ and $ A_0=0 $, $\tau_0=1$, and $ k=0 $;
		\State  Compute
		\begin{align*}
		& a_k=\frac{\tau_k+\sqrt{\tau_k^2+4\tau_k T A_k}}{2T},\quad  A_{k+1}=A_k+a_k, \\
		& \tau_{k+1}=\tau_k + a_k \mu, \quad \tilde{x}_k=\frac{A_ky_k+a_kx_k}{A_{k+1}};
		\end{align*}
        \label{line:repeat}
		\State Compute
		\begin{align*}
		y_{k+1}&:=\underset{u\in \bb R^d}\argmin\left\lbrace \gamma_k(u) + \frac{T}{2}\|u-\tilde{x}_k\|^2\right\rbrace, \\
		x_{k+1}&:= \underset{u\in \bb R^d}\argmin\left\lbrace a_k \gamma_k(u) + \frac{\tau_k}{2}\|u-x_k\|^2\right\rbrace, 
		\end{align*}
		where
		\[
			\gamma_k(u) := g(\tilde{x}_k) + \inner{\grad g(\tilde{x}_k)}{u-\tilde{x}_k} + \frac{B}{2}\|u-\tilde{x}_k\|^2;
		\]
	    \State Output $ \tilde{x}_k $ if $ \norm{\grad g(\tilde{x}_k)}  \le \sqrt{Ld} $; otherwise,set $ k \leftarrow k+1 $ and go to \Cref{line:repeat}.
    \end{algorithmic}
\end{algorithm}

% Note that there exists some $\lambda\in [0,1]$ such that $V(x) - V(y) = \inner{\grad V(\lambda x+(1-\lambda)y)}{x-y}$. Therefore, $\forall x,y \in \bb R^d$,
% \begin{align*}
%     \norm{V(x) - V(y) - \left \langle \grad V(y), x-y \right \rangle} &= \norm{\inner{\grad V(\lambda x+(1-\lambda)y) - \grad V(y)}{x-y}} \\
%     \mr{\Cref{assump:smooth}}&\le \frac{L}{2} \norm{x-y}^2.
% \end{align*}
The following lemma gives the convergence bound of the rejection sampling algorithm and \Cref{thm:rejection} is then a direct corollary.
\begin{lemma}[A corollary of Lemma~\ref{lem:AproxSta} and Proposition 3.4 in \cite{LC23}]\label{lem:reject} 
Assume $V$ is $L$-smooth. Then the random variable generated by Algorithm~\ref{alg:Rej} follows the distribution with density $\propto \exp\set{-V_y^\sigma(x)}$.
Moreover, if\;$\sigma^2 \le \frac{1}{2Ld}$, then the expected number of rejection steps in Algorithm~\ref{alg:Rej} is $\wt{\+O}(1)$. \footnote{The notation $\wt{\+O}$ subsumes a logarithmic term with regard to $L,d$ and $x^*_y = \arg\min_{x\in \bb R^d} V_y^{\sigma}(x)$, which is generated due to the call of \Cref{algo:AproxSta}.}
% \ctodo{$\wt{\+O}()$ here?}
\end{lemma}
\begin{algorithm}[H]
	\caption{RGO Rejection Sampling}
	\label{alg:Rej}
	\begin{algorithmic}[1]
	    \State Compute an approximate solution $w$ satisfying $\norm{\grad V_y^{\sigma} (w)} \le \sqrt{Ld}$ with Algorithm~\ref{algo:AproxSta}.  Let $h_y^w (x)= V(w) + \inner{\grad V(w)}{x-w} - \frac{L}{2}\|x -w\|^2 + \frac{1}{2\sigma^2}\norm{x -y}^2$
		\State Generate sample $X$ with density $\propto \exp \set{-h_{y}^w(x)}$ \label{line:repeat-in-rejection} 
		\State Generate sample $U\sim {\cal U}[0,1]$
		\If{$U \leq \frac{\exp \set{-V_y^\sigma(X)}}{\exp\set{-h_{y}^w(X)}}$}
            \State Output $X$
        \Else 
            \State Go to \Cref{line:repeat-in-rejection}
        \EndIf
	\end{algorithmic}
\end{algorithm}

% By combining Lemma~\ref{lem:AproxSta} and Lemma~\ref{lem:reject}, we can implement an algorithm that generates a sample from $\mu_{y,\sigma^2}$ with $\wt{\+O}(1)$ queries to $V$ and $\grad V$ in expectation. So we prove \Cref{thm:rejection}.
\section{Omitted proofs}\label{sec:lemmas}
\subsection{The proofs in \Cref{sec:prelim}}\label{sec:prelim-pf}

\begin{proof}[Proof of Proposition~\ref{prop:mLSIofSL}]
    By definition,
    \begin{align*}
        \+E_{\*P^{(T)}}(f,\log f) &= \frac{1}{2}\cdot \int_{\bb R^d\times \bb R^d} \tp{f(x)-f(y)}\tp{\log f(x)-\log f(y)} \mu(x)\*P^{(T)}(x,\dd y) \dd x \\
        \mr{\Cref{eq:MC}}&= \frac{1}{2} \E[X(T)\sim \xi_T]{\int_{\bb R^d} \int_{\bb R^d}(f(x)-f(y))(\log f(x)-\log f(y))\nu_T(x)\nu_T(y) \d x \d y}\\
        &=\E[X(T)\sim \xi_T]{ \int_{\bb R^d} f(x) \log f(x) \nu_T(x) \d x - \E[\nu_T]{f}\E[\nu_T]{\log f}}\\
        \mr{Jensen's inequality}&\geq \E[X(T)\sim \xi_T]{ \E[\nu_T]{f\log f} - \E[\nu_T]{f}\log\E[\nu_T]{ f}}\\
        &= \E[X(T)\sim \xi_T]{\Ent[\nu_T]{f}}.
    \end{align*}
    Therefore,
    \begin{align*}
        C^{\!{mLSI}}_{\mu}\tp{\*P^{(T)}} &= \sup_{f\colon \bb R^d\to \bb R_{> 0}} \frac{\Ent[\mu]{f}}{\+E_{\*P^{(T)}}(f,\log f)}
        \leq \sup_{f\colon \bb R^d\to \bb R_{>0}} \frac{\Ent[\mu]{f}}{\E[X(T)\sim \xi_T]{\Ent[\nu_T]{f}}}.
    \end{align*}
\end{proof}

To prove \Cref{thm:conservation}, we need the fact that the stochastic localization scheme is a linear-tilt scheme.
\begin{lemma}[Theorem 2 in \cite{EM22}]\label{lem:lineart}
    Write $L_x$ for the likelihood ratio process of $\nu_s$ with respect to $\mu$ at $x$, i.e., $L_s(x)\defeq \frac{\dd \nu_s}{\dd \mu}(x)$. Then there exists a Brownian motion $\set{W(s)}_{s\geq 0}$ adapted to the filtration generated by $\set{X(s)}_{s\geq 0}$, such that for all $x\in \bb R^d$, $s\geq 0$,
    \[
        \dd L_s(x) = L_s(x)\cdot \inner{x-\m{\nu_s}}{\dd W(s)}, \ L_0(x)=1,
    \]
    where $\dd L_s(x)$ denotes the time differential of $L_s(x)$.
\end{lemma}

\begin{proof}[Proof of \Cref{thm:conservation}]
In the following proof, all the differential notations refer to time derivatives. From Lemma~\ref{lem:lineart}, for any $f:\bb R^d\to \bb R$, 
\begin{align*}
    \dd \Var[\nu_s]{f} &= \int_{\bb R^d} f(x)^2 \cdot \dd  \nu_s(x) \dd x -  \dd  \tp{\int_{\bb R^d} f(x)  \nu_s(x) \dd x}^2\\
    \mr{Ito's Lemma}&= \int_{\bb R^d} f(x)^2 \cdot \nu_s(x)\cdot \inner{x-\m{\nu_s}}{\dd W(s)} \dd x \\
    &\quad - 2 \tp{\int_{\bb R^d} f(x)  \nu_s(x) \dd x}\cdot \tp{\int_{\bb R^d} f(x) \cdot \nu_s(x)\cdot \inner{x-\m{\nu_s}}{\dd W(s)} \dd x }\\
    &\quad - \tp{\int_{\bb R^d} f(x) \cdot \nu_s(x)\cdot \inner{x-\m{\nu_s}}{\dd W(s)} \dd x}^2.
\end{align*}
Let $\set{\+F_s}_{s\geq 0}$ be the filtration generated by $\set{X(s)}_{s\geq 0}$. Then $\E{\d W(s)\mid \+F_s} = 0$ and
\begin{align*}
    \E{\dd \Var[\nu_s]{f}\mid \+F_s} &= - \E{ \tp{\int_{\bb R^d} f(x) \cdot \nu_s(x)\cdot \inner{x-\m{\nu_s}}{\dd W(s)} \dd x}^2 \mid \+F_s }\\
    &= - \Var{\inner{\int_{\bb R^d} f(x) \cdot \nu_s(x)\cdot (x-\m{\nu_s}) \dd x}{\dd W(s)} \mid \+F_s }.
\end{align*}
Note that for any vector $v\in \bb R^n$ and any random vector $X$,
\begin{align*}
    \Var{\inner{v}{X}} & = \E{\inner{v}{X-\E{X}}^2}\\
    &= \sum_{1\leq i\leq n\atop 1\leq j\leq n} \E{v_iv_j\cdot (X_i-\E{X_i})(X_j-\E{X_j})}\\
    &= \inner{\!{Cov}(X)v}{v}.
\end{align*}
Therefore, 
\begin{align*}
    \E{\dd \Var[\nu_s]{f}\mid \+F_s}
    &= - \norm{\int_{\bb R^d} f(x) \cdot \nu_s(x)\cdot (x-\m{\nu_s}) \dd x}^2 \dd s\\
    &=-\sup_{\theta\in \bb R^n\atop \|\theta\|=1} \tp{\int_{\bb R^d} f(x) \cdot \nu_s(x)\cdot \inner{x-\m{\nu_s}}{\theta} \dd x}^2 \dd s\\
    &=-\sup_{\theta\in \bb R^n\atop \|\theta\|=1} \tp{\int_{\bb R^d} \tp{f(x) - \E[\nu_s]{f}} \cdot \nu_s(x)\cdot \inner{x-\m{\nu_s}}{\theta} \dd x}^2 \dd s\\
    \mr{Cauchy-Schwarz inequality}&\geq -\Var[\nu_s]{f}\cdot \sup_{\theta\in \bb R^n\atop \|\theta\|=1} \int_{\bb R^d}\inner{x-\m{\nu_s}}{\theta}^2 \nu_s(x)\dd x \dd s\\
    &= -\Var[\nu_s]{f}\cdot \sup_{\theta\in \bb R^n\atop \|\theta\|=1} \Var[\nu_s]{\inner{\theta}{X-\m{\nu_s}}} \dd s\\
    &= -\Var[\nu_s]{f}\cdot \sup_{\theta\in \bb R^n\atop \|\theta\|=1} \inner{\cov{\nu_s}\theta}{\theta}\dd s\\
    &=-\Var[\nu_s]{f}\cdot \norm{\cov{\nu_s}}_{\!{op}} \dd s.
\end{align*}
Consequently, we have
\begin{align*}
    \dd \E{ \Var[\nu_s]{f}} &= \E{ \dd \Var[\nu_s]{f}} \geq - \theta_s\cdot \E{\Var[\nu_s]{f}} \dd s,
\end{align*}
and 
\[
    \dd \log  \E{ \Var[\nu_s]{f}} \geq -\theta_s \dd s.
\]
Integrating both sides from $0$ to $T$,
\[
    \log \tp{\frac{\E{\Var[\nu_T]{f}}}{\Var[\mu]{f}}} \geq -\int_0^T \theta_s \dd s.
\]

% \begin{align*}
%     \dd \E{ \log \Var[\nu_s]{f}} &= \E{\dd \log \Var[\nu_s]{f}} \\
%     &= \E{\E{\dd \log \Var[\nu_s]{f}\mid \+F_s}}\\
%     &\geq - \E{\norm{\cov{\nu_s}}_{\!{op}}}\dd s \geq -\theta_s \dd s.
% \end{align*}
% \[
%     \E{\dd \log \Var[\nu_s]{f}} = \E{\E{\dd \log \Var[\nu_s]{f}\mid \+F_s}}\geq - \E{\norm{\cov{\nu_s}}_{\!{op}}}\dd s \geq -\theta_s \dd s.
% \]
% Integrating both sides from $0$ to $T$ and applying the Jensen's inequality, we have
% \[
%     \log \tp{\frac{\E{\Var[\nu_T]{f}}}{\Var[\mu]{f}}} \geq \E{ \log \tp{\frac{\Var[\nu_T]{f}}{\Var[\mu]{f}}}} \geq -\int_0^T \theta_s \dd s.
% \]

Then according to Proposition~\ref{prop:PIofSL},
\begin{align*}
    C^{\!{PI}}_\mu\tp{\*P^{(T)}} &= \sup_{f\colon \bb R^d\to \bb R} \frac{\Var[\mu]{f}}{\E[X(s)\sim \xi_T]{\Var[\nu_T]{f}}} \leq e^{\int_0^T \theta_s \dd s}.
\end{align*}
\end{proof}

\subsection{The proof in \Cref{sec:SLvsOU}}\label{sec:OUvsSL-pf}

\begin{proof}[Proof of Lemma~\ref{lem:OUvsSL}] 
    From \cref{eq:OU}, we know that $X^{\OU}(t)$ equals $e^{-t}\cdot X^{\OU}(0) + \sqrt{1-e^{-2t}}\cdot\zeta$ in distribution, where $\zeta\sim \+N(0,\Id)$ is independent with $X^{\OU}(0)$. For $X(s)$ defined in \cref{eq:SL},
    \[
        \sqrt{\frac{1}{s(1+s)}}\cdot X(s) = \sqrt{\frac{s}{(1+s)}} \cdot X + \sqrt{\frac{1}{s(1+s)}}\cdot B(s) 
    \]
    Since $\sqrt{\frac{s}{(1+s)}} = \sqrt{\frac{\frac{e^{-2t}}{1-e^{-2t}}}{1+\frac{e^{-2t}}{1-e^{-2t}}}} = e^{-t}$,$\frac{1}{\sqrt{1+s}} = \frac{1}{\sqrt{1+ \frac{e^{-2t}}{1-e^{-2t}}}} = \sqrt{1-e^{-2t}}$, the distributions of $X^{\!{OU}}(t)$ and $\sqrt{\frac{1}{s(1+s)}}\cdot X(s)$ are the same for any $t>0$.

    Recall that $\nu^{\OU}_t(\cdot\,|\,y)$ is the distribution of $X^{\!{OU}}(0)$ given $X^{\!{OU}}(t) = y$,then 
    \[
        \nu^{\OU}_t(x\,|\,y) \propto \mu(x) \cdot \exp\set{-\frac{1}{2(1-e^{-2t})}\norm{e^{-t}x - y}^2}.
    \]
    On the other hand, in the stochastic localization process, for any $s>0$ and $z\in \bb R^d$,
    \begin{align*}
        \nu_s(x\,|\,z) &\propto \mu(x)\cdot \exp\set{-\frac{\|sx-z\|^2}{2s}}\\
        % &\propto \mu(x)\cdot \exp\set{-\frac{\|sx-(z+v)\|^2}{2s} } \\
        \mr{$z=\sqrt{s(1+s)}y$} &= \mu(x)\cdot \exp\set{-\frac{\|sx-\sqrt{s(1+s)}y\|^2}{2s}} \\
        \mr{$s = \frac{e^{-2t}}{1-e^{-2t}}$ } &= \mu(x)\cdot \exp\set{-\frac{1}{2(1-e^{-2t})}\norm{e^{-t}x - y}^2}.
    \end{align*}
    % Let $p_y^{\prime}$ be the probability density function of the distribution $X$ given $X(s) = \sqrt{s(1+s)}y$. Since $\sqrt{\frac{1}{s(1+s)}} X(s) = e^{-t}X + \sqrt{1-e^{-2t}} \zeta^{\prime}$, $ p_y^{\prime}(x) \propto \mu(x) \times \exp\set{-\frac{1}{1-e^{-2t}}\norm{e^{-t}x - y}^2}$
    Therefore, $\nu_s(x\,|\,z) = \nu^{\OU}_t(x\,|\,y)$.
\end{proof}

\begin{proof}[Proof of Lemma~\ref{lem:tweedie}]
    In this proof, we may abbreviate $p_{X_s}$ as $p_s$. Note that 
\begin{align*}
    p_s(y) 
    = s^{-\frac{d}{2}}\int_{\bb R^d} \phi\tp{-\frac{\norm{y-sx}^2}{2s}}\cdot \mu(x) \d x
    = (2\pi \cdot s)^{-\frac{d}{2}} \int_{\bb R^d} \exp\tp{-\frac{\norm{y-sx}^2}{2s}}\mu(x)\d x.
\end{align*}
We let $q_s(x,y)\defeq \exp\tp{-\frac{\norm{y-sx}^2}{2s}}\mu(x)$ and have that
\begin{align*}
    \grad_y p_s(y) 
    &= (2\pi\cdot s)^{-\frac{d}{2}}\int_{\bb R^d} \exp\tp{-\frac{\norm{y-sx}^2}{2s}} \mu(x) \cdot \tp{x-y/s} \dd x\\
    &= (2\pi\cdot s)^{-\frac{d}{2}}\int_{\bb R^d} q_s(x,y) \cdot \tp{x-y/s} \dd x.
\end{align*}
It follows from Bayes' rule that $p_{X|X_s}(x|y) \propto q_s(x,y)$. Then we have
\begin{align*}
    \grad_y\log p_{s}(y) 
    = \frac{\grad_y p_{s}(y)}{p_{s}(y)}
    &=\frac{\int_{\bb R^d} q_s(x,y)\cdot (x-y/s) \d x}{\int_{\bb R^d} q_s(x,y)\d x}\\
    \mr{$p_{X|X_s}(x|y) = \frac{q_s(x,y)}{\int_{\bb R^d} q_s(x,y)\d x}$}
    &= \E[p_{X|X_s}(\cdot\,|\,y)]{X-y/s}.
\end{align*}
\end{proof}

%\ctodo{$\grad$ applying to a vector seems to uncommon. Need to define.}
\begin{proof}[Proof of Lemma~\ref{lem:hessian-pt}]
    In this proof we may also abbreviate $p_{X_s}$ as $p_s$. It follows from Lemma~\ref{lem:tweedie} that $\grad_y \log p_s(y) =  \E[p_{X|X_s}(\cdot\,|\,y)]{X} - y/s$. Let $Z_s(y)\defeq \int_{\bb R^d}q_s(x,y)\d x$, we have
    \begin{align}
        \grad_y^2 \log p_s(y)
        &= \grad_y \tp{Z_s(y)^{-1}\int_{\bb R^d} x\cdot q_s(x,y) \d x - y/s}\notag\\
        &=\int_{\bb R^d}x\cdot \grad_y\tp{\frac{q_s(x,y)}{Z_s(y)} }^\top\d x - \frac{\Id}{s}.\label{eq:hessian}
    \end{align}
    Direct calculations show that
    \begin{align*}
        \grad_y\tp{\frac{q_s(x,y)}{Z_s(y)}}
        &= \frac{1}{Z_s(y)}\cdot \grad_y q_s(x,y) - \frac{q_s(x,y)}{Z_s(y)^2}\cdot \grad_y Z_s(y)\\
        \mr{$\grad_y q_s(x,y) = q_s(x,y) (x-y/s)$}
        &=\frac{q_s(x,y)}{Z_s(y)}\cdot \tp{x-y/s-\frac{\grad_y Z_s(y)}{Z_s(y)}}\\
        \mr{$\frac{\grad_y Z_s(y)}{Z_s(y)} = \E[p_{X|X_s=y}(\cdot|y)]{X}-y/s$}
        &= \frac{q_s(x,y)}{Z_s(y)}\cdot \tp{x - \E[p_{X|X_s}(\cdot|y)]{X}}.
    \end{align*}
    Plugging this into \cref{eq:hessian}, we immediately have
    \[
        \grad_y^2\log p_s(y) = \cov{p_{X|X_s}(\cdot|y)} - \frac{\Id}{s}.
    \]
\end{proof}

\subsection{The proofs in \Cref{sec:algo}}\label{sec:algo-pf}
% \htodo{(1) modify $eps_0$ in lem:OUconverge; (2) delete the proof there; (3) delete the convergence of OU in prelim; (4) modify the value of $s_0$.}
Before proving Lemma~\ref{lem:OUconverge}, we first prove the convergence of the OU process in Proposition~\ref{prop:OUconvergence}, which is a commonly known result (see, e.g., \cite{BGL14,VW19,HZD24}).
\begin{proposition}[Convergence of the OU process]\label{prop:OUconvergence}
    Suppose $\mu$ satisfies Assumption~\ref{assump:smooth} and \ref{assump:moment}. Then
    \[
        \DKL{\xi^{\OU}_t}{\+N(0,\Id)} \leq e^{-2t}\cdot (Ld+M).
    \]
\end{proposition}

\begin{proof}[Proof of Proposition~\ref{prop:OUconvergence}]
    Without loss of generality, we assume $\int e^{-V(x)}dx = 1$ for brevity in the following proof. From Lemma~\ref{lem:sobolev},
    \begin{align*}
        \DKL{\mu}{\+N\tp{0, \Id}} & \le \frac{1}{2} \int e^{-V(x)} \cdot \norm{\grad \log \frac{e^{-V(x)}}{\frac{1}{(2\pi)^{d/2}} e^{-\frac{\norm{x}^2}{2}}}}^2 dx \\ 
        &=\frac{1}{2} \int e^{-V(x)} \cdot \norm{\grad V(x) +x }^2 dx \\ 
        % &\le \frac{1}{2} \int e^{-V(x)} (\norm{\grad V(x) +x }^2 + \norm{\grad V(x) - x}^2) dx 
        &\le \int e^{-V(x)} \norm{\grad V(x)}^2 dx + \int e^{-V(x)} \norm{x}^2 dx
    \end{align*}
    For a matrix $A\in \bb R^{d\times d}$, let $\text{Tr}(A)$ denote its trace. Then for the first term, 
    \begin{align*}
        \int  e^{-V(x)} \norm{\grad V(x)}^2 dx &= \int e^{-V(x)}\cdot  \grad\cdot \grad V(x) dx \\
        &= \int e^{-V(x)} \Delta V(x) dx \\
        &=\int e^{-V(x)} \cdot \text{Tr}(\grad^2 V(x)) dx \\ 
        \mr{\Cref{assump:smooth}}&\le Ld
    \end{align*}
    For the second term, $\int e^{-V(x)} \norm{x}^2 dx = \E[X \sim \mu]{\norm{x}^2} \le M$.
    Thus 
    \[
        \DKL{\mu}{\+N\tp{0, \Id}} \le Ld + M.
    \]
    % Let $p^{\!OU}_t(x)$ be the density function of $X^{\!OU}_t$.

    According to the Fokker-Planck equation of the OU process, 
    \begin{equation}
        \partial_t \xi^{\!OU}_t(x) = \grad \cdot (\xi^{\!OU}_t(x) x ) + \Delta \xi^{\!OU}_t(x) =  \grad \cdot \left (\xi^{\!OU}_t(x) \grad \log{\frac{\xi^{\!OU}_t(x)}{\frac{1}{(2\pi)^{d/2}}\exp \set{-\frac{1}{2}\norm{x}^2}}}\right). \label{eq:OU-1}
    \end{equation}
    Therefore
    \begin{align*}
        \frac{\dd}{\dd t} \DKL{\xi_{t}^{\!OU}}{\+N\tp{0, \Id}} &= \frac{\dd}{\dd t} \int  \xi^{\!OU}_t(x) \log{\frac{ \xi^{\!OU}_t(x)}{\frac{1}{(2\pi)^{d/2}} \exp\set{-\frac{1}{2}\norm{x}^2}}} dx\\ 
        &= \int \partial_t \left( \xi^{\!OU}_t(x) \log{\frac{ \xi^{\!OU}_t(x)}{\frac{1}{(2\pi)^{d/2}} \exp \set{-\frac{1}{2}\norm{x}^2}}}\right)dx  \\
        &= \int \partial_t  \xi^{\!OU}_t(x)\cdot  \log{\frac{ \xi^{\!OU}_t(x)}{\frac{1}{(2\pi)^{d/2}} \exp \set{-\frac{1}{2}\norm{x}^2}}} dx +  \int \partial_t  \xi^{\!OU}_t(x) \dd x\\
        \mr{\Cref{eq:OU-1}}&=  \int \left( \grad  \cdot \left ( \xi^{\!OU}_t(x) \grad \log{\frac{ \xi^{\!OU}_t(x)}{\frac{1}{(2\pi)^{d/2}}\exp \set{-\frac{1}{2}\norm{x}^2}}}\right)  \right)\cdot \log{\frac{ \xi^{\!OU}_t(x)}{\frac{1}{(2\pi)^{d/2}} \exp\set{-\frac{1}{2}\norm{x}^2}}} dx  \\
        \mr{integration by parts} &= - \int  \xi^{\!OU}_t(x) \cdot \norm{\grad \log \frac{ \xi^{\!OU}_t(x)}{\frac{1}{(2\pi)^{d/2}} \exp\set{-\frac{1}{2}\norm{x}^2}}}^2 dx \\
        \mr{Lemma~\ref{lem:sobolev}} &\le -2 \DKL{\xi_{t}^{\!OU}}{\+N\tp{0, \Id}},
    \end{align*}
    and consequently,
    \[
        \frac{\dd }{\dd t} \log{\DKL{\xi_{t}^{\!OU}}{\+N\tp{0, \Id}}} \le -2.
    \]
    Integrating the two sides of the above equation from $0$ to $t$, we get the result directly.
% thus \[
% \DKL{\xi_{t}^{\!OU}}{\+N\tp{0, \Id}} \le e^{-\frac{t}{2}} \cdot (Ld +M)
% \]   
\end{proof}

Then Lemma~\ref{lem:OUconverge} is a corollary of Proposition~\ref{prop:OUconvergence}.
\begin{proof}[Proof of Lemma~\ref{lem:OUconverge}]
    From Lemma~\ref{lem:OUvsSL}, $\xi_{s_0}(x)  = \tp{s_0(1+s_0)}^{-\frac{d}{2}}\cdot \xi^{\OU}_{t_0}(y)$ for $t_0=\log\sqrt{\frac{1+s_0}{s_0}}= \frac{1}{2}\log\frac{2(Ld+M)}{\eps^2}$ and $y= \frac{1}{\sqrt{s_0(1+s_0)}}x$.
    By definition,
    \begin{align*}
        \DKL{\xi_{s_0}}{\+N\tp{0, s_0(1+s_0)\Id}} &= \int_{\bb R^d} \xi_{s_0}(x) \cdot \log\tp{\frac{\xi_{s_0}(x)}{\tp{2\pi\cdot s_0(1+s_0)}^{-\frac{d}{2}} \cdot e^{-\frac{\|x\|^2}{2s_0(1+s_0)}}}} \dd x \\
        \mr{ $y= \frac{1}{\sqrt{s_0(1+s_0)}}x$ }&= \int_{\bb R^d} \xi^{\OU}_{t_0}(y)\cdot \log \frac{\xi^{\OU}_{t_0}(y)}{\tp{2\pi}^{-\frac{d}{2}} \cdot e^{-\frac{\|y\|^2}{2}}}\dd y \\
        &= \DKL{\xi^{\OU}_{t_0}}{\+N\tp{0, \Id}}\\
        \mr{Proposition~\ref{prop:OUconvergence}}&\leq \frac{\eps^2}{2}.
    \end{align*}
    The result then follows from the Pinsker's inequality.
\end{proof}

\subsection{The proofs in \Cref{sec:concatenation}}\label{sec:concat-pf}
\begin{proof}[Proof of Lemma~\ref{lem:beforeconcate}]
    We slightly abuse the notation by letting $\xi'_s(\cdot\,|\, y)$ be the distribution of $Z(s)$ in \cref{eq:modifiedSL} given $X_{s_0}=y$.
    We first prove that for any $w\in \bb R^d$, $\xi_{T}(w) = \int_{\bb R^d} \xi_{s}(y)\cdot \xi'_{T-s}(w-y\,|\,y) \dd y$. By the definition of $\xi'_{T-s}$,
    \begin{align}
        &\phantom{{}={}}\int_{\bb R^d} \xi_{s}(y)\cdot \xi'_{T-s}(w-y\,|\,y) \dd y \notag \\
        &= \int_{\bb R^d} \xi_{s}(y)\cdot \int_{\bb R^d} \nu_s(x\,|\,y) \cdot \frac{\exp\set{-\frac{\norm{w-y-(T-s)x}^2}{2(T-s)}}}{\tp{2\pi (T-s)}^{\frac{d}{2}}} \dd x\dd y  \notag \\
        &=\int_{\bb R^{d}\times \bb R^{d}} \mu(x)\cdot \frac{\exp\set{-\frac{\norm{y-sx}^2}{2s}}}{(2\pi s)^{\frac{d}{2}}}\cdot \frac{\exp\set{-\frac{\norm{w-y-(T-s)x}^2}{2(T-s)}}}{\tp{2\pi (T-s)}^{\frac{d}{2}}} \dd x \dd y \notag \\
        &=\int_{\bb R^{d}\times \bb R^{d}}  \frac{\mu(x)}{\tp{2\pi s \cdot 2\pi (T-s)}^{\frac{d}{2}}} \cdot \exp\set{ \frac{s(T-s)}{2T}\cdot \norm{x+\frac{w-(T-s)x}{T-s}}^2 -\frac{s\|x\|^2}{2} - \frac{\norm{w-(T-s)x}^2}{2(T-s)}  \right.  \notag \\
        &\qquad \qquad \qquad \qquad \left. - \frac{T}{2s(T-s)}\cdot \norm{y - \frac{s(T-s)\cdot\tp{x+\frac{w-(T-s)x}{T-s}}}{T}}^2} \dd x \dd y \notag \\
        &= \int_{\bb R^{d}\times \bb R^{d}}  \frac{\mu(x)}{\tp{2\pi T}^{\frac{d}{2}}} \cdot \exp\set{ \frac{s(T-s)}{2T}\cdot \norm{x+\frac{w-(T-s)x}{T-s}}^2 -\frac{s\|x\|^2}{2} - \frac{\norm{w-(T-s)x}^2}{2(T-s)}} \dd x \notag \\
        & = \int_{\bb R^{d}\times \bb R^{d}}  \frac{\mu(x)}{\tp{2\pi T}^{\frac{d}{2}}} \cdot \exp\set{-\frac{\|w-Tx\|^2}{2T}} \dd x  \notag \\
        &= \xi_{T}(w). \label{eq:beforeconcate-1}
    \end{align}
    Then from Lemma~\ref{lem:RGO},
    \begin{align*}
        \E[\xi_s]{\E[\xi'_{T-s}]{\Var[\nu'_{T-s}]{f}}} &= \int_{\bb R^d\times \bb R^d} \xi_s(y)\cdot \xi'_{T-s}(z\,|\,y)\cdot \Var[\nu_{T}( \cdot\,|\,y+z)]{f} \dd z \dd y\\
        \mr{$z=w-y$} &= \int_{\bb R^d\times \bb R^d} \xi_s(y)\cdot \xi'_{T-s}(w-y\,|\,y)\cdot \Var[\nu_{T}( \cdot\,|\,w)]{f} \dd w \dd y \\
        \mr{\Cref{eq:beforeconcate-1}}&= \int_{\bb R^d} \xi_T(w) \cdot \Var[\nu_{T}( \cdot\,|\,w)]{f} \dd w \\
        &= \E[\xi_T]{\Var[\nu_T]{f}}.
    \end{align*}
    The proof of the entropy case is similar.
\end{proof}

\begin{proof}[Proof of Lemma~\ref{lem:hessian-2}]
    Let $q_s(x,y) = \pi(x)\cdot e^{-f(x,y)}$. Then $\xi(y) \propto \int_{\bb R^d} q_s(x,y) \dd x$ and
    \[
        -\grad_y \log \xi(y) = \frac{\int_{\bb R^d}\grad_y f(x,y)\cdot q_s(x,y) \dd x}{\int_{\bb R^d} q_s(x,y) \dd x}.
    \]
    Therefore,
    \begin{align*}
        -\grad^2_y \log \xi(y) &= \frac{\int_{\bb R^d}\tp{\grad^2_y f(x,y) - \grad_y f(x,y)^{\otimes 2}}\cdot q_s(x,y) \dd x}{\int_{\bb R^d} q_s(x,y) \dd x} \\
        &\quad + \frac{\tp{\int_{\bb R^d}\grad_y f(x,y)\cdot q_s(x,y) \dd x}^{\otimes 2}}{\tp{\int_{\bb R^d} q_s(x,y) \dd x}^2}\\
        &= \E[X\sim \nu_y]{\grad^2_y f(X,y) - \grad_y f(X,y)^{\otimes 2}} + \E[X\sim \nu_y]{\grad_y f(X,y)}^{\otimes 2}\\
        &= \E[X\sim \nu_y]{\grad^2_y f(X,y)} - \Cov[X\sim \nu_y]{\grad_y f(X,y)}.
    \end{align*}
\end{proof}

\section{Technical lemmas}

%  The following lemma comes from \cite{HZD24}
\begin{lemma}[Lemma F.4 in \cite{HZD24}] \label{lem:sobolev}
Consider an $m$-strongly log-concave distribution $\pi$. For any distribution $\zeta$, we have 
\[
    \DKL{\zeta}{\pi} \le \frac{1}{2m} \int \zeta(x) \norm{\grad \log{\frac{\zeta(x)}{\pi(x)}}}^2 dx
\] 
\end{lemma}
% \begin{proof}[Proof of \Cref{lem:sobolev}]

%     From classical theory, we know that the $m$-strongly log-concave distribution $p^*$ implies the following log-sobolev inequality 
%     \[
%          \Ent[p^*]{f^2} \le \frac{2}{m} \int p^*(x) \norm{\grad f(x)}^2  dx
%     \]
%     Let $f(x) = \sqrt{\frac{p(x)}{p^*(x)}}$, then 
%     \[
%         \Ent[p^*]{f^2} = \Ent[p^*]{\frac{p(x)}{p^*(x)}} = \DKL{p}{p^*}
%     \]
%     \begin{align*}
%         \frac{2}{m} \int p^*(x) \norm{\grad \sqrt{\frac{p(x)}{p^*(x)}}}^2 dx &= \frac{2}{m} \int p(x) \sqrt{\frac{p^*(x)}{p(x)}}^2 \norm{\grad \sqrt{\frac{p(x)}{p^*(x)}}}^2 dx \\ &= \frac{2}{m} \int p(x) \norm{\grad \log{\sqrt{\frac{p(x)}{p^*(x)}}}}^2 dx \\ &= \frac{1}{2m} \int p(x) \norm{\grad \log{\frac{p(x)}{p^*(x)}}}^2 dx
%     \end{align*}
%     Thus 
%     \[
%         \DKL{p}{p^{*}} \le \frac{1}{2m} \int p(x) \norm{\grad \log{\frac{p(x)}{p^*(x)}}}^2 dx
%     \]
% \end{proof}

\begin{lemma}[Concentration of Poisson distribution]\label{lem:poistail}
    Let $X\sim \!{Pois}(\lambda)$. Then for any $s>0$, $\Pr{X\geq \lambda +s}\leq e^{-\frac{s^2}{2(\lambda + s)}}$.
\end{lemma}
\begin{proof}
    From the Markov's inequality, for some fixed $\theta>0$
    \[
        \Pr{X\geq \lambda +s} = \Pr{e^{\theta X}\geq e^{\theta(\lambda +s)}} \leq \E{e^{\theta X}} \cdot e^{-\theta(\lambda +s)}.
    \]
    It is a standard result that the moment generating function of $X$ is $\E{e^{\theta X}} = e^{\lambda\tp{e^{\theta}-1}}$. Choosing $\theta = \log \tp{1+\frac{s}{\lambda}}$, we have
    \[
        \Pr{X\geq \lambda +s}\leq e^{\lambda\tp{e^{\theta}-1} - \theta(\lambda +s)} = e^{s - (\lambda +s)\log \tp{1+\frac{s}{\lambda}}} = e^{-\frac{s^2}{2\lambda}\cdot 2\tp{\tp{\frac{1}{u^2} + \frac{1}{u}}\log (1+u) - \frac{1}{u}}},
    \]
    where $u=\frac{s}{\lambda}$.
    It remains to prove $2\tp{\tp{\frac{1}{u^2} + \frac{1}{u}}\log (1+u) - \frac{1}{u}} \geq \frac{1}{1+u}$. Let
    \[
        g(u) = 2(u+1)\cdot \tp{\tp{\frac{1}{u^2} + \frac{1}{u}}\log (1+u) - \frac{1}{u}}.
    \]
    Note that
    \begin{align*}
        g(u) &= 2(u+1)\cdot \frac{(1+u)\cdot \log (1+u) - u}{u^2} \\
        &= 2(u+1)\cdot \frac{(1+u)\cdot \tp{u - u^2/2 + o(u^2)} - u}{u^2}\\
        &\overset{u\to 0}{\longrightarrow} 1,
    \end{align*}
    and for any $u\in (0,\infty)$,
    \begin{align*}
        g'(u) &= 2\tp{\tp{\frac{1}{u^2} + \frac{1}{u}}\log (1+u) - \frac{1}{u}} + 2(u+1)\cdot \tp{\tp{\frac{1}{u^2} + \frac{1}{u}}\cdot \frac{1}{u+1} - \tp{\frac{2}{u^3} + \frac{1}{u^2}}\cdot \log (1+u) + \frac{1}{u^2}}\\
        &= 2\tp{ -\tp{\frac{2}{u^3} + \frac{2}{u^2}}\cdot \log (1+u) + \frac{2}{u^2} + \frac{1}{u}} \\
        &= \frac{2}{u^3}\cdot \tp{2u+u^2 - 2(1+u)\log (1+u)}.
    \end{align*}
    Let $h(u) = 2u+u^2 - 2(1+u)\log (1+u)$. We have $h(0)=0$ and $h'(u) = 2u - 2\log(1+u)\geq 0$. Therefore, $g'(u)>0$ for any $u>0$.

    Therefore, we have $g(u)>1$ for all $u>0$ and 
    \[
        \Pr{X\geq \lambda +s} \leq \exp\set{-\frac{s^2}{2\lambda}\cdot \frac{1}{1+\frac{s}{\lambda}}\cdot g(u)} \leq e^{-\frac{s^2}{2(\lambda + s)}}.
    \]
\end{proof}

\begin{lemma}[The $\chi^2$ divergence between Gaussians]\label{lem:chi2ofGaussian}
    Consider two multi-variate Gaussian distribution $\mu_1=\+N(x,\sigma^2\cdot\Id)$ and $\mu_2=\+N(y,\sigma^2\cdot\Id)$. Then $\chi^2(\mu_1\,\|\,\mu_2) = e^{\frac{\|x-y\|^2}{\sigma^2}} -1$.
\end{lemma}
\begin{proof}
    By the definition of $\chi^2$ divergence, we have
    \begin{align*}
        \chi^2(\mu_1\,\|\,\mu_2) &= \int_{\bb R^d} \frac{\mu_1(z)^2}{\mu_2(z)} \dd z -1\\
        &= \frac{1}{(2\pi\sigma^2)^{\frac{d}{2}}}\cdot \int_{\bb R^d} \exp\set{-\frac{\norm{z-x}^2}{\sigma^2} + \frac{\norm{z-y}^2}{2\sigma^2} } \dd z -1 \\
        &= \frac{1}{(2\pi\sigma^2)^{\frac{d}{2}}}\cdot \int_{\bb R^d} \exp\set{- \frac{\norm{z-(2x-y)}^2}{2\sigma^2} - \frac{\|x\|^2}{\sigma^2} + \frac{\|y\|^2}{2\sigma^2} + \frac{\|2x-y\|^2}{2\sigma^2}} \dd z -1 \\
        &= e^{\frac{\|x-y\|^2}{\sigma^2}} -1.
    \end{align*}
\end{proof}

\begin{lemma}\label{lem:chi4ofGaussian}
    For any $x,y\in \bb R^d$,
    \[
      \int \frac{\+N(z;x,\sigma^2\cdot\Id)^4}{\+N(z;y,\sigma^2\cdot\Id)^3}\d z=e^{\frac{6\|x-y\|^2}{\sigma^2}}.
    \]
\end{lemma}
\begin{proof}
    We have
    \begin{align*}
        \int \frac{\+N(z;x,\sigma^2\cdot\Id)^4}{\+N(z;y,\sigma^2\cdot\Id)^3}\d z &= \frac 1{(2\pi\sigma^2)^{\frac d2}}\int_{\bb R^d} \exp\set{-\frac{2\|z-x\|^2}{\sigma^2}+\frac{3\|z-y\|^2}{2\sigma^2}}\d z\\
        &= \frac 1{(2\pi\sigma^2)^{\frac d2}}\int_{\bb R^d} \exp\set{-\frac {\|z-(4x-3y)\|^2}{2\sigma^2}+\frac{6\|x-y\|^2}{\sigma^2}}\d z\\
        &= e^{\frac{6\|x-y\|^2}{\sigma^2}}.
    \end{align*}
\end{proof}

% Consider two probability space $(\Omega,\+F,\+P)$ and $(\Omega,\+F,\+Q)$. Let $X,Y$ be two random variables defined on $\Omega$. Here we slightly abuse the notation and let $\+P_X$, $\+P_Y$ and $\+P_{X,Y}$ be the distribution of $X$, $Y$ and the joint distribution of $X,Y$ on $(\Omega,\+F,\+P)$. Define similar notations for $\+Q$. The following two lemmas are the well-known data processing inequality and chain rule for total variation distance. For the completeness, we also incorporate their proofs here.
% \begin{lemma}[Data processing inequality for total variation distance]
%     \[
%         \DTV(\+P_X,\+Q_X)\leq \DTV(\+P_{X,Y},\+Q_{X,Y}).
%     \]
% \end{lemma}
% \begin{proof}
    
% \end{proof}

% \begin{lemma}[Chain rule for total variation distance]
%     \[
%         \DTV(\+P_{X,Y},\+Q_{X,Y}) \leq \DTV(\+P_{Y},\+Q_{Y}) + \E[Y\sim \+Q_Y]{\DTV(\+P_{X|Y},\+Q_{X|Y})}.
%     \]
% \end{lemma}
% \begin{proof}
    
% \end{proof}

\end{document}